\documentclass[a4paper,10pt]{article}
\usepackage{amsmath,amsfonts,amssymb,amsthm}
\usepackage{latexsym}
\usepackage{color}

\usepackage{graphicx} 
\usepackage{caption}
\usepackage{subcaption}

\usepackage{hyperref}
\usepackage{enumerate}

\usepackage[margin=2.5cm]{geometry}

\def \cbb{\mathbb{C}}

\def \nbb{\mathbb{N}}
\def \rbb{\mathbb{R}}

\def \zbb{\mathbb{Z}}

\def \bcal {\mathcal{B}}
\def \ccal {\mathcal{C}}
\def \dcal {\mathcal{D}}

\def \hcal {\mathcal{H}}
\def \ical {\mathcal{I}}

\def \ocal {\mathcal{O}}
\def \pcal {\mathcal{P}}
\def \qcal {\mathcal{Q}}
\def \rcal {\mathcal{R}}
\def \scal {\mathcal{S}}

\def \wcal {\mathcal{W}}

\def \bfk  {\mathfrak{B}}

\def \xfk  {\mathfrak{X}}

\newcommand{\thetav}{\boldsymbol{\theta}}
\newcommand{\etav}{\boldsymbol{\eta}}
\newcommand{\lambdav}{\boldsymbol{\lambda}}

\newcommand{\zetav}{\boldsymbol{\zeta}}
\newcommand{\piv}{\boldsymbol{\pi}}
\newcommand{\xiv}{\boldsymbol{\xi}}

\newcommand{\xv}{\boldsymbol{x}}
\newcommand{\yv}{\boldsymbol{y}}

\newcommand{\zerov}{\boldsymbol{0}}

\def \. { \,\! }

\def\clap#1{\hbox to 0pt{\hss#1\hss}}

\def\mathclap{\mathpalette\mathclapinternal}

\def\mathclapinternal#1#2{\clap{$\mathsurround=0pt#1{#2}$}}

\def \cdotarg { \, \cdot \, }

\DeclareMathOperator{\im}{Im}
\DeclareMathOperator{\re}{Re}

\newcommand{\idop}{\boldsymbol{1}}

\def \st {^\ast}

\DeclareMathOperator{\supp}{supp}
\DeclareMathOperator{\dom}{dom}

\def \expltext#1 {\\ \text{\footnotesize{ (#1) }}\\}
\def \intercomm#1 {\\ \text{\footnotesize{ (#1) }}\\}
\def \undercomm#1 {\underset{\text{\scriptsize{ (#1) }}}}
\def \overcomm#1 {\overset{\text{\scriptsize{ (#1) }}}}

\newcommand{\Hil}{\mathcal{H}}

\newcommand{\fpno}{\Hil^{\omega,\mathrm{f}}}
\newcommand{\fpnp}{P^{\mathrm{f}}}

\newcommand{\boundedops}{\bfk(\Hil)}

\newcommand{\qf}{\mathcal{Q}}

\newcommand{\A}{\mathcal{A}}

\newcommand{\hscalar}[2]{\langle #1 , #2 \rangle }
\newcommand{\bighscalar}[2]{\big\langle  #1 , #2 \big\rangle }

\newcommand{\gnorm}[2]{\lVert #1 \rVert_{#2}}

\newcommand{\onorm}[2]{\lVert #1 \rVert_{#2}^\omega}

\newcommand{\zd}{z^\dagger}

\newcommand{\cmeA}[2]{f_{#1}^{\lbrack #2 \rbrack}}
\newcommand{\cmeB}[2]{f_{#1}{\lbrack #2 \rbrack}}
\newcommand{\cme}[2]{\mathchoice{\cmeA{#1}{#2}}{\cmeB{#1}{#2}}{\cmeB{#1}{#2}}{\cmeB{#1}{#2}}}

\newcommand{\affiliated}{\mathrel{\eta}}

\DeclareMathOperator{\sign}{sign}

\DeclareMathOperator*{\res}{res}

\DeclareMathOperator*{\domain}{dom}

\newcommand{\uidx}[1]{^{\lbrack #1 \rbrack}}
\newcommand{\lidx}[1]{_{ #1 }}

\newcommand{\even}{^\mathrm{even}}
\newcommand{\odd}{^\mathrm{odd}}

\newcommand{\psp}[1]{\pi_{#1}}

\newcommand{\inv}{_\mathrm{I}}

\newcommand{\cmpqed}{}
\newcommand{\cmponly}[1]{}
\newcommand{\ahponly}[1]{}

\numberwithin{equation}{section}

\newtheorem{definition}{Definition}[section]
\newtheorem{lemma}[definition]{Lemma}
\newtheorem{proposition}[definition]{Proposition}
\newtheorem{theorem}[definition]{Theorem}
\newtheorem{corollary}[definition]{Corollary}

\newenvironment{acknowledgements}{\paragraph{Acknowledgements}}{}

\title{Towards an explicit  construction  of local observables \\{}in integrable quantum field theories}
\author{Henning Bostelmann\thanks{%
University of York, Department of Mathematics, York YO10 5DD, United Kindom. 
E-mail: \href{mailto:henning.bostelmann@york.ac.uk}{\nolinkurl{henning.bostelmann@york.ac.uk}}
}
 \and Daniela Cadamuro\thanks{%
Universit\"at Leipzig, Institut f\"ur Theoretische Physik, Br\"uderstra\ss{}e 16, 04103 Leipzig, Germany.
E-mail: \href{mailto:daniela.cadamuro@itp.uni-leipzig.de}{\nolinkurl{daniela.cadamuro@itp.uni-leipzig.de}}
}
}

\date{2 January 2020}

\begin{document}

\maketitle

\begin{abstract}
We present a new viewpoint on the construction of pointlike local fields in integrable models of quantum field theory. As usual, we define these local observables by their form factors; but rather than exhibiting their $n$-point functions and verifying the Wightman axioms, we aim to establish them as closed operators affiliated with a net of local von Neumann algebras, which is defined indirectly via wedge-local quantities. We also investigate whether these fields have the Reeh-Schlieder property, and in which sense they generate the net of algebras.
Our investigation focuses on scalar models without bound states. We establish sufficient criteria for the existence of averaged fields as closable operators, and complete the construction in the specific case of the massive Ising model.
\end{abstract}

\ahponly{\maketitle} 


\section{Introduction}

Quantum field theory is based on the concept of local observables, i.e., operators associated with points or regions of space-time which commute at spacelike distances. Yet these local objects are notoriously difficult to construct in the presence of interaction. Even in simplified situations that are amenable to a mathematically rigorous treatment, explicit control over the local observables is hard to obtain: one can either make a direct ansatz for quantum fields, but face difficulties in controlling their singular nature, particularly in the high-energy regime; or one can define local quantities via an abstract limiting process, allowing one to control their functional analytic properties, but losing track of their explicit form. 

These difficulties are exemplified in the models we consider in this article, namely, quantum integrable models on 1+1 dimensional Minkowski space; open questions remain about the structure of their local observables, despite substantial research focusing on this issue. 

There are two complementary approaches to obtaining local observables in integrable models. The first of them, known as the \emph{form factor program} \cite{Smirnov:1992,BabujianFoersterKarowski:2006}, aims at constructing point-local quantum fields $\Phi(x)$ directly.
Their $n$-point functions are expanded in a series by inserting a basis of intermediate asymptotic states, for example for $n=2$,
\begin{equation}\label{eq:phiintro}
   \hscalar{\Omega}{\Phi(x)\Phi(y)\Omega} \ahponly{\!} = \ahponly{\!}  \sum_{m=0}^\infty \int \frac{d\theta_1 \dotsm d\theta_m}{(2\pi)^m} \big\lvert \langle \Omega |\Phi(0) | \theta_1,\dotsc,\theta_m \rangle^{\mathrm{in}} \big\rvert^2 e^{i(y-x)\sum_{j=1}^m p(\theta_j)},
\end{equation}
where the $\theta_j$ are rapidities. The expansion terms $\langle \Omega |\Phi(0) | \theta_1,\dotsc,\theta_m \rangle^{\mathrm{in}}$ are called \emph{form factors;} locality and covariance requirements for $\Phi(x)$ then lead to restrictions on these, the \emph{form factors equations}. For specific forms of the interaction, such as the massive Ising model \cite{SchroerTruong:1978} or the sinh-Gordon model \cite{FringMussardoSimonetti:1993}, one can find explicit solutions of the form factor equations. The remaining problem is now to control the convergence of the infinite series, in order to verify, e.g., the Wightman axioms \cite{StrWig:PCT}. However, only partial results in certain asymptotic regimes exist so far, even in the simplest interacting case, the massive Ising model \cite{BabujianKarowski:2004}.

The second approach \cite{SchroerWiesbrock:2000-1,Lechner:review}, which we call the \emph{operator algebraic} one, proceeds in an indirect way: One first constructs observables with weaker localization properties, namely, quantum fields localized in spacelike wedges. While not the desired final result, these wedge-local fields can explicitly be described and mathematically controlled. Passing to algebras of bounded operators $\A(\wcal)$ associated with wedges $\wcal$, one then obtains observable algebras in bounded regions by taking intersections: Where a bounded region is the intersection of two wedges, $\ocal=\wcal_1 \cap \wcal_2$, one sets  
\begin{equation}\label{eq:aintro}
 \A(\ocal) := \A(\wcal_1) \cap \A(\wcal_2).
\end{equation}
This net of algebras quite directly fulfills the Haag-Kastler axioms \cite{Haa:LQP}. The mathematically hard task, however, is to show nontriviality of the intersections. This can be done by abstract arguments in a class of models \cite{Lechner:2008,AL2017}, including the sinh-Gordon and Ising models, at least for sufficiently large regions $\ocal$. But explicit control of the form of these observables $A \in \A(\ocal)$ is lost; essentially, they are obtained from the axiom of choice.

Thus, known results allow one to either control the explicit form of observables or their functional analytic behaviour. In the present paper, we propose a method to close this gap using a hybrid approach: We take our local observables to be defined by explicit expressions for pointlike fields, following ideas from the form factor program. Then, we aim to show that they are local operators in a mathematically strict sense: namely, that their closures are affiliated with the algebras $\A(\ocal)$ as defined in \eqref{eq:aintro}. Relying on affiliation with von Neumann algebras rather than on $n$-point functions of fields gives us the flexibility needed to tackle longstanding convergence issues. 

We carry out this programme in the context of scalar integrable quantum field theories without bound states. In this context, we present sufficient criteria that make this approach work, and that do not refer to details of the interaction, i.e., to the two-particle scattering function.  We verify the criteria in the massive Ising model. 

To that end, we make use of techniques from \cite{BostelmannCadamuro:characterization} which exhibit the connection between the two approaches to integrable systems. The local operators $A \in \A(\ocal)$ constructed abstractly in \cite{Lechner:2008} can be expanded into a series,
\begin{equation}\label{eq:expansionintro}
  A = \sum_{m,n=0}^\infty \int \frac{d^m \theta \, d^n \eta}{m!n!} F_{m+n}(\thetav+i\zerov,\etav+i\piv-i\zerov) z^{\dagger}(\theta_1)\cdots z^\dagger(\theta_m) z(\eta_1)\cdots z(\eta_n),
\end{equation}
where $z,\zd$ are ``interacting'' annihilators and creators (cf.~\cite{Lashkevich:1994}), and $F_{m+n}$ are meromorphic functions, paralleling the form factor program. In fact, they fulfill very similar relations to the known form factor equations, plus certain growth bounds encoding the localization of $A$. (These will be recalled in Sec.~\ref{sec:locality}.) The expansion \eqref{eq:expansionintro} is not restricted to bounded operators, but should also hold for other local quantities, such as locally averaged quantum fields, or more general quadratic forms $A$. However, most functional analytic properties (such as boundedness or closability) of the operator $A$ are not directly visible on the level of the expansion coefficients $F_k$, and the series exists only in the sense of matrix elements between finite particle number states, where the sum is actually finite. Locality for these objects is only defined in a weak sense, namely, as relative locality to the wedge-local field mentioned above ($\omega$-locality, see Def.~\ref{definition:omegalocal} below). 

Hence our main line of argument is as follows. As our input, we take meromorphic functions $F_k$ that fulfill a refined version of the form factor axioms (see Thm.~\ref{theo:conditionFD} below); in concrete models, candidates are known in the literature. This gives us our observables (averaged quantum fields) as quadratic forms by \eqref{eq:expansionintro}. Additionally, we assume a certain summability condition for the series \eqref{eq:expansionintro}, resulting in our local fields as closed operators. Based on the locality conditions for the functions $F_k$, the operators are then shown to be affiliated with the local algebras $\A(\ocal)$. We note that this construction does not depend on a priori information on the size of the algebras $\A(\ocal)$.

We verify our summability condition in an example, the massive Ising model. In our context, the massive Ising model is the 1+1-dimensional massive integrable quantum field theory with constant two-particle S-matrix $S=-1$. While the \emph{massless} Ising model is generated by (the even powers of) a free Fermi field, the \emph{massive} Ising model differs from this in important aspects: its scattering states are bosonic, and its PCT operator is different from the related Fermi field on the same Hilbert space \cite[Sec.~II]{BostelmannCadamuroFewster:2013}. It is, in this sense, a theory of interacting Bosons, even if with a very simple type of interaction. While quadratic expressions in the Fermi field generate a subnet of $\A$, this is a \emph{proper} subnet, and $\A(\ocal)$ contains also operators with odd particle number transfer. Crucially, for these the series in \eqref{eq:expansionintro} cannot terminate, thus providing us with a test case for our ideas. The Ising model has been constructed in the operator algebraic context \cite{Lechner:2005} and as a Euclidean quantum field theory \cite{PalmerTracy:ising}, but to the authors' knowledge, direct convergence results for the series \eqref{eq:expansionintro} in Minkowski space are new.

We stress that, while the observables we construct are formally averaged versions of the local field of the form factor program, given as $A=\Phi(g) = \int d^2x \,g(x) \Phi(x)$ with $\Phi(x)$ as in \eqref{eq:phiintro}, we do \emph{not} claim that they fulfill the Wightman axioms. For one, we do not use Schwartz functions $g$, but rather functions of Jaffe class \cite{Jaffe:1967}; but this is a more minor point. More fundamentally, we do not want to, or need to, control the product of two such operators; we do not claim that their $n$-point functions exist, or that the fields have a common invariant domain. For our interpretation as local observables, it is sufficient to show that $\Phi(g)$ is affiliated with $\A(\ocal)$ where $\supp g \subset \ocal$.

In a slight extension of scope, one can ask whether this method leads to \emph{all} local observables of the model. Namely, for each bounded region $\ocal$ of space-time, we obtain a linear space $\qf(\ocal)$ of quadratic forms (which extend to closed operators, affiliated with $\A(\ocal)$); this set would also include the ``composite fields'' or ``descendant operators'' of the model, although we do not explicitly deal with normal products or product expansions. Is this $\qf(\ocal)$ maximally large, in a well-defined sense? One criterion would be whether the space has the Reeh-Schlieder property, i.e., whether $\qf(\ocal)\Omega$ is dense in the Hilbert space of the model. A somewhat stricter notion is whether the elements of $\qf(\ocal)$, or their spectral data, generate the algebra $\A(\ocal)$. Both questions can be traced back to sufficient conditions on the functions $F_k$, where for the last mentioned point, we understand ``generate'' in the sense of the dual of a net of algebras. We also investigate which consequences this completeness has for the net $\A(\ocal)$ itself.

The paper is organized as follows. In Sec.~\ref{sec:background}, we recall our mathematical setting, including the definition of wedge-local algebras and the characterization of local operators in terms of a series expansion. Then, in Sec.~\ref{sec:criterion}, we develop sufficient criteria for closability of operators, affiliation with local algebras, and completeness in the sense of the Reeh-Schlieder property or duality. 

We explicitly treat the situation in the Ising model in Secs.~\ref{sec:evenexamples} and \ref{sec:oddexamples}. In the Ising model, for local observables with even particle number transfer, the series \eqref{eq:expansionintro} can be finite, whereas for odd particle number transfer, it is necessarily infinite. We discuss the easier, even case in Sec.~\ref{sec:evenexamples}, hoping it will be instructive for the reader. The odd case is treated in Sec.~\ref{sec:oddexamples}; it involves quite delicate estimates of the singular integral operators with kernels $F_{m+n}(\thetav+i\zerov,\etav+i\piv-i\zerov)$, which are boundary values of meromorphic functions, with first-order poles located on the boundary.

We summarize our results, and give an outlook on future work, in Sec.~\ref{sec:outlook}. Two appendices provide technical results needed in Sec.~\ref{sec:oddexamples}: Appendix \ref{app:poly} deals with symmetric Laurent polynomials that are required for treating composite fields, and Appendix \ref{app:modd} investigates the singularity structure of a certain multivariable meromorphic function needed in the construction.

This paper is partly based on one of the authors' Ph.D.\ thesis \cite{Cadamuro:2012}.


\section{Background}\label{sec:background}

The context of this paper are integrable models of quantum field theory on 1+1 dimensional Minkowski space, with a single species of massive scalar particle. We also exclude bound states, i.e., the two-particle scattering function will not have poles in the physical strip. (For possible generalizations, see Sec.~\ref{sec:othermodels}.) We formulate them in the mathematical framework of \cite{Lechner:2008,BostelmannCadamuro:expansion,BostelmannCadamuro:characterization}, the relevant aspects of which we now recall.

\subsection{Hilbert space}
The model under discussion is specified by the mass $\mu>0$ of the particle and the \emph{scattering function} $S$, a meromorphic function on $\cbb$ which is bounded on the strip $0 \leq \im \zeta \leq \pi$ and fulfills the symmetry relations
\begin{equation}\label{eq:srelat}
  \quad S(\zeta)^{-1}=S(-\zeta)=\overline{S(\bar{\zeta})\vphantom{\hat S}}=S(\zeta+i \pi).
\end{equation}
Given these, we define a modified Fock space $\hcal := \bigoplus_{n=0}^\infty \hcal_n$, where $\Hil_n$ is the ``$S$-symmetric part'' of $L^2(\rbb^n,d\thetav)$, i.e., consists of wave functions $\psi_n$ (depending on rapidity arguments) which behave under transposition of variables according to
\begin{equation}
   \psi_n(\theta_1,\dotsc,\theta_j,\theta_{j+1},\dotsc,\theta_n) = S(\theta_{j+1}-\theta_{j}) \psi_n(\theta_1,\dotsc,\theta_{j+1},\theta_{j},\dotsc,\theta_n).
\end{equation}
We denote the projector in $\Hil$ onto $\Hil_n$ as $P_n$, and set $\fpnp_k:=\sum_{n=0}^k P_n$. On $\Hil$, we have a representation $U$ of the proper Poincar\'e group, under which translations and boosts $U(x,\Lambda)$ and spacetime reflections $U(j)$ act on  $\psi\in\Hil_n$ as
\begin{align}
  (U(x,\Lambda)\psi)(\thetav)&= e^{ i p(\thetav) \cdot x } \psi(\theta_1  -\lambda, \dotsc,\theta_n-\lambda),
\\
   (U(j)\psi)(\thetav) &= \overline{\psi(\theta_{n},\ldots,\theta_{1})}.
\end{align}
\sloppy
Here $\lambda$ is the boost parameter of $\Lambda$, and $p(\thetav) = \sum_{k=1}^{n}p(\theta_{k})$, with $p(\theta) = \mu (\cosh \theta,\sinh\theta)$. We will also denote $E(\thetav)=p^0(\thetav)/\mu$.
This $U$ is an \mbox{(anti-)}unitary, strongly continuous, positive-energy representation, and the translations have the Fock vacuum $\Omega$ as their unique invariant vector (up to scalar factors). The generator of time translations will be denoted $H$.

On $\Hil$, annihilation and creation operators $z(\theta)$ and $z^\dagger(\theta)$ act, defined as usual in a distributional sense on finite particle number vectors, but fulfilling an $S$-deformed version of the CCR \cite[Sec.~3]{Lechner:2008}. The ``smearing functions'' of these operator-valued distributions will often be Fourier transforms of functions $f \in \scal(\rbb^2)$, taken with the convention 
\begin{equation}
f^\pm(\theta) :=\frac{1}{2\pi}\int d^2 x \, f(x)e^{\pm ip(\theta)\cdot x}.
\end{equation}

\subsection{Quadratic forms}

We wish to describe observables as operators or quadratic forms on $\Hil$ with a certain high-energy behaviour. To that end, let $\omega:[0,\infty)\to [0,\infty)$ be an \emph{analytic indicatrix} \cite[Def.~2.1]{BostelmannCadamuro:characterization}, that is, a function growing slightly less than linearly with certain additional conditions; here we just note that
\begin{align}
\label{eq:omegabeta}
 &&\omega(p) &= \beta\log(1+p) \quad
& \text{for some $\beta>0$},
\\ \label{eq:omegaalpha}
\text{and} &&
\omega(p)&= p^{\alpha}
& \text{for some $\alpha\in(0,1)$}
\end{align}
are valid examples. Associated with $\omega$ and an open region $\ocal\subset \rbb^2$, we define the test function space
\begin{equation}\label{eq:domega}
\begin{aligned}
 \dcal^\omega(\ocal) := \{ \ahponly{&} f \in \dcal(\ocal) : 
 \ahponly{\\  & \quad} \theta \mapsto e^{\omega(\cosh \theta)} f^\pm(\theta) \text{ is bounded and square integrable} \}.
\end{aligned}
\end{equation}
In the case \eqref{eq:omegabeta}, this $\dcal^\omega(\ocal)$ is identical to $\dcal(\ocal)=\mathcal{C}_c^\infty(\ocal)$, whereas in the case \eqref{eq:omegaalpha}, it is a dense subspace.
We also consider the dense subspace $\fpno \subset \Hil$ of vectors $\psi$ such that $\gnorm{e^{\omega(H/\mu)}\psi}{} < \infty$ and which have finite particle number ($\fpnp_n\psi=\psi$ for some $n$). Further, let $\qf^\omega$ be the space of quadratic forms $A$ on $\fpno \times \fpno$ such that the norms
\begin{equation}\label{eq:aomeganorm}
 \gnorm{A}{n}^{\omega} := \frac{1}{2} \gnorm{\fpnp_n A e^{-\omega(H/\mu)}\fpnp_n}{} + \frac{1}{2} \gnorm{\fpnp_n e^{-\omega(H/\mu)} A \fpnp_n}{}
\end{equation}
are finite for any $n \in \nbb_0$. Examples of such forms are smeared normal-ordered monomials in the annihilators and creators \cite[Prop.~2.1]{BostelmannCadamuro:expansion}, written in formal integral notation as
\begin{equation}\label{eq:zzf}
   z^{\dagger m} z^n(f) = \int d^m\theta \,d^n\eta \, f(\thetav,\etav) z^{\dagger}(\theta_1)\cdots z^\dagger(\theta_m) z(\eta_1)\cdots z(\eta_n),
\end{equation}
where $f\in \dcal(\rbb^{m+n})'$ is such that the following norm $\onorm{f}{m \times n}$ is finite:
\begin{align}
\label{eq:crossnorm2}
\onorm{f}{m \times n} &:= \frac{1}{2} \gnorm{(\thetav,\etav)\mapsto e^{-\omega(E(\thetav))}f(\thetav,\etav)}{m \times n} 
\ahponly{ \notag \\ &\hphantom{:=} } + \frac{1}{2} \gnorm{(\thetav,\etav)\mapsto f(\thetav,\etav)e^{-\omega(E(\etav))}}{m \times n},
\\
\label{eq:crossnorm1}
\gnorm{f}{m \times n} &:= \sup \Big\{ \big\lvert \! \int f(\thetav,\etav) g(\thetav) h(\etav)  d^m\theta d^n\eta \,\big\rvert :
\ahponly{ \notag \\ & \qquad \qquad \qquad }
   g \in \dcal(\rbb^m), \, h \in \dcal(\rbb^n), \,  \gnorm{g}{2} \leq 1, \, \gnorm{h}{2} \leq 1\Big\}.
\end{align}
In fact, \emph{all} $A \in \qf^\omega$ can be decomposed into monomials of the form \eqref{eq:zzf}: One can find distributions $\cme{m,n}{A}$  such that \cite[Thm.~3.8]{BostelmannCadamuro:expansion}
\begin{equation}\label{eq:expansionrecall}
  A = \sum_{m,n=0}^\infty \int \frac{d^m \theta \, d^n \eta}{m!n!} \cme{m,n}{A}(\thetav,\etav) z^{\dagger}(\theta_1)\cdots z^\dagger(\theta_m) z(\eta_1)\cdots z(\eta_n).
\end{equation}
The sum is finite in matrix elements, so that convergence issues do not arise at this point.
Vice versa, given distributions $f_{m, n}$ such that $\onorm{f_{m,n}}{m \times n}<\infty$, we can define $A \in \qf^\omega$ by the sum above. 
For an explicit expression of the (unique) relation between $A$ and $\cme{m,n}{A}$, see \cite[Sec.~3.1]{BostelmannCadamuro:expansion}. The symmetry representation $U$ acts on $\qcal^\omega$ by adjoint action, and correspondingly on the expansion coefficients $\cme{m,n}{A}$; we refer to \cite[Sec.~3.3]{BostelmannCadamuro:expansion} for details.

\subsection{Locality}\label{sec:locality}

We now describe locality of our observables in open spacetime regions $\rcal$, the most relevant being: the right wedge $\wcal$ with tip at the origin; its causal complement, the left wedge $\wcal'$; shifted wedges $\wcal_x$, $\wcal_x'$ with tip at $x$; double cones $\ocal_{x,y}=\wcal_x \cap \wcal_y'$; and the standard double cone $\ocal_r$ of radius $r$ around the origin. We start by introducing the wedge-local field \cite[Sec.~3]{Lechner:2008} 
\begin{equation}\label{eq:wlfield}
\phi(f):= \zd(f^{+}) + z(f^{-}), \qquad 
 f\in \mathcal{S}(\mathbb{R}^{2}).
\end{equation}
This field, or its formal kernel $\phi(x)$, can with respect to the symmetry representation $U$ be consistently interpreted as localized in the wedge $\wcal_x'$. We then define a von Neumann algebra of bounded operators associated with the right wedge as 
\begin{equation}
\A(\wcal) = \{e^{i\phi(f)^-} \,|\, f \in \dcal_\rbb^\omega(\wcal') \}'.
\end{equation}
(The subscript ${}_\rbb$ indicates real-valuedness. In \cite{Lechner:2008} this was introduced with $\omega=0$, but the algebra is actually independent of $\omega$ by density arguments for the test functions $f$.) From here, algebras associated with other wedges $\wcal_x$ and $\wcal_y'$ can be defined by symmetry transformations, and for double cones $\ocal_{x,y} = \wcal_x \cap \wcal_y'$ via $\A(\ocal_{x,y}):=\A(\wcal_x) \cap \A(\wcal_y')$. In this way, one obtains a Haag-Kastler net $\A$ of local algebras for every region of Minkowski space, where the vacuum $\Omega$ is cyclic and separating for $\A(\wcal)$, Haag duality for wedges holds, i.e., $\A(\wcal_x')=\A(\wcal_x)'$, and the Tomita-Takesaki modular group of $\A(\wcal)$ coincides with the boosts $U(0,\Lambda)$. It is a priori not clear whether the algebras $\A(\ocal_{x,y})$ contain any operator other than multiples of the identity, but under certain conditions (``modular nuclearity''),\footnote{Note that our analysis in the following will not rely on the modular nuclearity condition.}
the vacuum is in fact cyclic for these as well \cite[Sec.~2]{Lechner:2008}. 

This gives a well-defined sense of locality for bounded operators. For (unbounded) quadratic forms, the situation is different, as we cannot formulate commutation relations between these directly. Instead, we can define a weaker notion by means of relative locality to the wedge-local field $\phi$:
\begin{definition} \label{definition:omegalocal}
\cite[Def.~2.4]{BostelmannCadamuro:characterization}
  Let $A \in \qf^\omega$. We say that $A$ is \emph{$\omega$-local} in $\wcal_x$ iff%
  \footnote{%
    The ``commutator'' $[A,\phi(f)]$ is actually well-defined in matrix elements for quadratic forms $A \in \qf^\omega$, since $\phi(f)$ and $\phi(f)^\ast$ map $\fpno$ into $\fpno$.
  }
\begin{equation}\label{eq:acommute}
     [A,\phi(f)] = 0
\quad
\text{for all }
    f \in \dcal^\omega(\wcal_x'), \text{ as a relation in $\qf^\omega$.%
}
\end{equation}
$A$ is called $\omega$-local in $\wcal_x'$ iff $U(j) A^\ast U(j)$ is $\omega$-local in $\wcal_{-x}$. $A$ is called $\omega$-local in the double cone $\ocal_{x,y} = \wcal_x \cap \wcal'_y$ iff it is $\omega$-local in both $\wcal_x$ and $\wcal'_y$.
\end{definition}
We will clarify in Sec.~\ref{sec:criterion} how $\omega$-locality is related to the local net $\A$, as well as to locality conditions for closed unbounded operators.

For our purposes, it is crucial to know how locality of $A\in\qf^\omega$ is reflected in the properties of its expansion coefficients $\cme{m,n}{A}$. In fact, for $A$ localized in a double cone, one finds that $\cme{m,n}{A}$ are distributional boundary values of meromorphic functions $F_{m+n}$ at specific points. To formulate this, consider the regions in $\rbb^k$,
\begin{equation}\label{eq:ik}
\begin{aligned}
 \ical^k_+ &= \{\lambdav : 0 < \lambda_1 < \ldots < \lambda_k < \pi \}, \\
 \ical^k_- &= \{\lambdav : -\pi < \lambda_1 < \ldots < \lambda_k < 0 \}.
\end{aligned}
\end{equation}
When we write boundary distributions of the type $F_k(\thetav+i(0,\dotsc,0), \etav+i(\pi-0,\dotsc,\pi-0))$, or $F_k(\thetav+i \zerov, \etav + i\piv - i\zerov)$ for short, this is understood as an approach from within the region $\ical^k_+$, and similar for $\ical^k_-$. With this, we can characterize $\omega$-locality in the double cone $\ocal_r$ as follows, in a reformulation of \cite[Theorem~5.4]{BostelmannCadamuro:characterization}.

\begin{theorem}\label{theo:conditionFD}

Let $\omega$ be an analytic indicatrix and let $r>0$. Let $F=(F_{k})_{k=0}^\infty$ be a collection of functions $\cbb^k \to \bar\cbb$ which fulfills the following conditions\footnote{The conditions \ref{it:fdmero}--\ref{it:fdrecursion} coincide with the properties of form factors, e.g., in \cite{BabujianFoersterKarowski:2006,BK:02} by setting, in the notation of \cite{BabujianFoersterKarowski:2006}, $F^{\mathcal{O}}(\zeta_1, \ldots, \zeta_k) = (2\sqrt{\pi})^k F_k(\zeta_k, \ldots, \zeta_1)$ and suppressing the indices and matrices related to particle species, which are absent in our context.} for any fixed $k$, and with $\zetav \in \cbb^k$ arbitrary:

\begin{enumerate}
\renewcommand{\theenumi}{(FD\arabic{enumi})}
\renewcommand{\labelenumi}{\theenumi}

\item \label{it:fdmero}
\emph{Analyticity:}
$F_k$ is meromorphic on $\cbb^k$, and analytic where $\im \zeta_1 < \ldots < \im \zeta_k < \im \zeta_1 + \pi$.

\item \label{it:fdsymm} \emph{$S$-symmetry:} 
$
\displaystyle{
F_k(\zetav)
=  S(\zeta_{j+1}-\zeta_j) F_k(\zeta_1,\dotsc,\zeta_{j+1},\zeta_j,\dotsc,\zeta_k)
}$
for any $1 \leq j < k$.

\item \label{it:fdperiod} \emph{$S$-periodicity:}
$\displaystyle{
F_k (\zeta_1,\dotsc,\zeta_{k-1},\zeta_k+2\pi i) =  F_k (\zeta_k, \zeta_1,\dotsc,\zeta_{k-1} ).
}$

\item \label{it:fdrecursion}

\emph{Recursion relations:}
The $F_k$ have first order poles at $\zeta_n-\zeta_m = i \pi$, where $1 \leq m < n \leq k$,
and
\begin{equation*}
\res_{\zeta_2-\zeta_1 = i \pi} F_{k}( \zetav )
= - \frac{1}{2\pi i }
\Big(1-\prod_{j=1}^{k} S(\zeta_1-\zeta_j) \Big)
F_{k-2}( \zeta_3,\dotsc,\zeta_k ).
\end{equation*}

\item \label{it:fdboundsreal}
\emph{Bounds at nodes:}
For each $j \in \{0,\dotsc,k\}$, we have
\begin{equation*}
\onorm{ F_k\big( \cdotarg + (\underbrace{i \zerov}_{\mathclap{j\;\mathrm{entries}}},i\piv-i\zerov) }{j \times (k-j)} < \infty, \quad
\onorm{ F_k\big( \cdotarg - (\underbrace{i\piv -i \zerov}_{j\;\mathrm{entries}}, i\zerov)}{j \times (k-j)} < \infty.
\end{equation*}

\item \label{it:fdboundsimag}
\emph{Pointwise bounds:}
There exist $c,c'>0$ such that for all $\zetav\in\rbb^k+i \ical^k_\pm$:
\begin{equation*}
  |F_k(\zetav)| \leq c \,{ \operatorname{dist}(\im \zetav,\partial \ical_\pm^k)^{-k/2}} \prod_{j=1}^k \exp \big(\mu r  |\im \sinh \zeta_j|+ c' \omega(\cosh \re \zeta_j)\big).
\end{equation*}
\end{enumerate}

\noindent
Then, the unique quadratic form $A \in \qf^\omega$ fulfilling 
\begin{equation}\label{eq:fmnboundary}
\cme{m,n}{A}(\thetav,\etav) = F_{m+n}(\thetav+i\zerov, \etav + i\piv - i\zerov) 
\end{equation}
is $\omega$-local in $\ocal_r$.

Conversely, let $A \in \qf^\omega$ be $\omega$-local in $\ocal_r$. Then there exists (uniquely) a family of functions $(F_k)$ which fulfill conditions \ref{it:fdmero}--\ref{it:fdboundsimag} above such that 
\eqref{eq:fmnboundary} holds.
\end{theorem}


\section{Locality of operators and quadratic forms}\label{sec:criterion}

In this paper, we investigate local (unbounded) \emph{operators} in integrable models, going beyond the quadratic forms considered earlier \cite{BostelmannCadamuro:expansion,BostelmannCadamuro:characterization}. More specifically, we aim at closed operators affiliated with the local von Neumann algebras $\A(\ocal)$. This class, while still technically manageable, seems large enough to contain a variety of accessible examples, including smeared pointlike fields where they exist \cite{FreHer:pointlike_fields,Wollenberg:1985}. The present section gives general criteria that allow us to investigate the problem, independent of the scattering function $S$ and of specific examples of local observables. The criteria will later be applied to examples in the case $S=-1$, in Secs.~\ref{sec:evenexamples} and \ref{sec:oddexamples}.

We first clarify in Sec.~\ref{sec:critclosability} how quadratic forms in $\qf^\omega$ relate to closed (unbounded) operators, and establish sufficient criteria for convergence of the infinite series \eqref{eq:expansionrecall} in this context. Then, in Sec.~\ref{sec:critlocality}, we show how the closed operators are related to \emph{bounded} local operators, in the sense of affiliation with the local algebras. Lastly, in Sec.~\ref{sec:reehschlieder}, we ask when a set of quadratic forms is large enough to describe \emph{all} local observables of the quantum field theory, in the sense of the Reeh-Schlieder property and of generating the net of local von Neumann algebras.

Throughout the section, an analytic indicatrix $\omega$ is kept fixed.

\subsection{Closable operators and summability}  \label{sec:critclosability}

We will be concerned with the extension of quadratic forms $A\in\qf^\omega$ to closed operators. Since $A$ is a priori only a quadratic form, we clarify in which case this extension, or closure, is to be understood. 

\begin{definition} \label{def:closable}
 $A \in \qf^\omega$ is called \emph{$\omega$-closable} if there exists a closed operator $A^-$, with $\fpno \subset \dom A^- \cap \dom (A^{-})^\ast$ and for which $\fpno$ is a core, such that $A^-$ coincides with $A$ as a quadratic form on $\fpno \times \fpno$.
\end{definition}

Correspondingly, the operator $A^-$, which is uniquely determined, is called the \emph{$\omega$-closure} of $A$. (It may depend on $\omega$, but this will not matter for our purposes.) A simple criterion for $\omega$-closability is as follows.

\begin{lemma}\label{lemma:closable}
  $A\in\qf^\omega$ is $\omega$-closable if, and only if, the expression $\hscalar{ \psi }{ A \chi}$ has a continuous linear extension to $\chi \in \Hil$ for any fixed $\psi \in \fpno$, and a continuous antilinear extension to $\psi \in \Hil$ for any fixed $\chi \in \fpno$.
\end{lemma}

\begin{proof}
    Let $A \in \qf^\omega$. The two continuity conditions imply that $A$ can be extended to a linear operator $A_0:\fpno\to\Hil$ such that also $\fpno\subset\dom A_0^\ast$. In particular, $A_0$ and $A_0\st$ are both densely defined, which implies that $A^- := A_0^{\ast\ast}$ is a closed extension of $A$ with core $\fpno$ \cite[Thm.~VIII.1]{Reed:1972}. Also, $(A^-)\st = A_0^{\ast\ast\ast}=A_0\st$, which is defined at least on $\fpno$. Hence $A$ is $\omega$-closable. The converse is evident.
\cmpqed\end{proof}

In particular, this shows that the $\omega$-closable elements form a subspace of $\qf^\omega$. While the criterion in Lemma~\ref{lemma:closable} is easy to state, it is rather hard to apply in examples where the expansion coefficients $\cme{m,n}{A}$ are used to define $A$. We will therefore deduce a \emph{sufficient} criterion for $\omega$-closability which is based directly on estimates for the $\cme{m,n}{A}$. The idea is to establish absolute convergence of the series \eqref{eq:expansionrecall} in a certain sense. (In the sense of quadratic forms, the series is always well-defined as it is finite in matrix elements; for obtaining closed operators, however, convergence issues become relevant.)

\begin{proposition}\label{proposition:summable1}
 Let $A \in \qf^\omega$. Suppose that for each fixed $n$,
\begin{equation}\label{eq:summable1}
   \sum_{m=0}^\infty \frac{2^{m/2}}{\sqrt{m!}} \Big( \onorm{ \cme{m,n}{A} }{m \times n} + \onorm{ \cme{n,m}{A} }{n \times m} \Big) < \infty.
\end{equation}
Then, $A$ is $\omega$-closable.
\end{proposition}

\begin{proof}
By \cite[Prop.~2.1]{BostelmannCadamuro:expansion}, the annihilator-creator monomials fulfill the estimate, $k \in \nbb_0$,
\begin{equation}
\big\|  z^{\dagger m} z^n(f)e^{- \omega(H/\mu)} \fpnp_k \big\| \leq 2 \frac{\sqrt{k!(k-n+m)!}}{(k-n)!}\|f\|^{\omega}_{m\times n}.
\end{equation}
Using this estimate in the expansion \eqref{eq:expansionrecall}, we obtain for $\psi,\chi\in\fpno$,
\begin{equation}
 | \hscalar{\psi}{A  \chi } |
  \leq \gnorm{\psi}{} \gnorm{e^{\omega(H/\mu) } \chi}{} \sum_{m=0}^\infty \sum_{n=0}^k \frac{2}{m!n!} \frac{\sqrt{k!(k-n+m)!}}{(k-n)!} \onorm{ \cme{m,n}{A} }{m \times n} ,
\end{equation}
where $k=k(\chi)$ can be chosen independent of $\psi$. Estimating $k!/n!(k-n)! \leq 2^k$ and $(k-n+m)!/(k-n)!m! \leq 2^{k-n+m}$, we obtain
\begin{equation}\label{eq:aqknorm}
 | \hscalar{\psi}{A  \chi } | 
  \leq 2^{k+1} \gnorm{\psi}{} \gnorm{e^{\omega(H/\mu) } \chi}{}  
   \sum_{n=0}^k  \frac{2^{-n/2}}{\sqrt{n!}} \sum_{m=0}^\infty  \frac{2^{m/2}}{\sqrt{m!}} \onorm{ \cme{m,n}{A} }{m \times n} ,
\end{equation}
 which converges by assumption. Thus the matrix element is $\hcal$-continuous in $\psi$ at fixed $\chi$. A similar argument, with the roles of $m$ and $n$ exchanged, shows continuity in $\chi$ at fixed $\psi$. The result then follows from Lemma~\ref{lemma:closable}.
\cmpqed\end{proof}

Under a stricter summability condition, we can deduce an additional property that will become relevant later, in Proposition~\ref{proposition:locality}(\ref{it:aclosed}).

\begin{proposition}\label{proposition:summable2}
 Let $A \in \qf^\omega$. Suppose that
\begin{equation}\label{eq:summable2}
   \sum_{m,n=0}^\infty \frac{2^{(m+n)/2}}{\sqrt{m!n!}} \onorm{ \cme{m,n}{A} }{m \times n} < \infty.
\end{equation}
Then, $A$ is $\omega$-closable; and for any $g \in \dcal^\omega(\rbb^2)$, we have 
$\exp (i \phi(g)^-) \fpno \subset \domain{A^-}$.
\end{proposition}

\begin{proof}
In view of Proposition~\ref{proposition:summable1}, only the second part requires proof. We recall the estimates \cite[Eq.~(2.18)]{BostelmannCadamuro:expansion}
\begin{equation}\label{eq:omegaz}
\begin{aligned}
 \gnorm{ e^{\omega(H/\mu)} \zd(f) e^{-\omega(H/\mu)} \fpnp_\ell  }{}
  &\leq \sqrt{\ell+1} \, \gnorm{e^{\omega(\cosh {\cdot})}f}{},
\\
 \gnorm{ e^{\omega(H/\mu)} z(f) e^{-\omega(H/\mu)} \fpnp_\ell  }{}
  &\leq \sqrt{\ell} \, \gnorm{f}{}.
\end{aligned}
\end{equation}
With $\phi(g) = \zd(g^+) + z(g^-)$, it follows that for any $j \in \nbb$,
\begin{equation}
 \gnorm{ e^{\omega(H/\mu)} \phi(g)^j e^{-\omega(H/\mu)} \fpnp_\ell  }{}
  \leq  \sqrt{\frac{(\ell+j)!}{\ell!}}\,c_g^j , \quad 
\end{equation}
where $c_g := \gnorm{e^{\omega(\cosh {\cdot})} g^+}{} + \gnorm{g^-}{}$.
Since $\fpno$ consists of analytic vectors for $\phi(g)$, and since $\phi(g)$ changes the particle number by at most 1, we then obtain for $k \geq \ell$,
\begin{multline}
 \gnorm{ P_k e^{\omega(H/\mu)} e^{i \phi(g)^-} e^{-\omega(H/\mu)} \fpnp_\ell  }{}
  \leq \sum_{j=k-\ell}^\infty \frac{1}{j!}  \gnorm{ e^{\omega(H/\mu)} \phi(g)^j e^{-\omega(H/\mu)} \fpnp_\ell  }{}
\\
\leq c_g^{k-\ell} \frac{2^{k/2}}{\sqrt{(k-\ell)!}} \sum_{p=0}^\infty \frac{(\sqrt{2} c_g)^p}{\sqrt{p!}}
\leq c_g' \frac{(\sqrt{2}c_g)^k}{\sqrt{(k-\ell)!}}
\end{multline}
with some constant $c_g'>0$ depending on $g$.
For any $\chi \in \fpno$, we therefore have with suitable $\ell$,
\begin{multline}\label{eq:aexpest}
 \gnorm{ A^- P_k e^{i \phi(g)^-} \chi  }{}
\ahponly{\\}
 \leq \gnorm{ A^-  e^{-\omega(H/\mu)}  \fpnp_k }{}
 \gnorm{ P_k e^{\omega(H/\mu)} e^{i \phi(g)^-} e^{-\omega(H/\mu)} \fpnp_\ell  }{}
 \gnorm{  e^{\omega(H/\mu)} \chi }{}
\\
  \leq
 2 c_g' \gnorm{ e^{\omega(H/\mu)} \chi}{}
\frac{(\sqrt{8}c_g)^k}{\sqrt{(k-\ell)!}}
\sum_{m,n=0}^\infty
\frac{2^{(m-n)/2}}{\sqrt{m!n!}} \onorm{ \cme{m,n}{A} }{m \times n},
\end{multline}
where the estimate on $A^-$ has been deduced from \eqref{eq:aqknorm}. The series on the r.h.s.\ exists by hypothesis. 

Now set $\hat \chi_k := \fpnp_k \exp(i \phi(g)^-) \chi \in \fpno$. Since the r.h.s.~of \eqref{eq:aexpest} is summable over $k$, both $\hat \chi_k$ and $A^- \hat \chi_k$ are convergent sequences in $\Hil$. As $A^-$ is closed, this implies that $\lim_{k\to\infty} \hat \chi_k = \exp(i \phi(g)^-) \chi$ is contained in the domain of $A^-$.
\cmpqed\end{proof}

\subsection{Locality}  \label{sec:critlocality}

We now consider \emph{local} observables of our model. In Sec.~\ref{sec:background}, we introduced two notions of locality: a net of von Neumann algebras $\A(\ocal)$, where locality can be expressed in terms of commutation relations in the usual sense, and the concept of $\omega$-locality for quadratic forms (Def.~\ref{definition:omegalocal}), which was based on relative locality to the wedge-local fields $\phi(g)$, $\phi'(g)$. A priori, $\omega$-locality is a much weaker notion, since it only involves commutators in the weak sense between a restricted set of observables. However, we show that for suitably regular quadratic forms (bounded or $\omega$-closable), $\omega$-local observables can be linked to the net of local algebras.

\begin{proposition}\label{proposition:locality}
In the following, let $\rcal$ be one of the regions $\wcal_x$, $\wcal_y'$, $\ocal_{x,y}$ for some $x,y\in\rbb^2$.
\begin{enumerate}[(a)]

\item \label{it:abounded}
  Let $A$ be a bounded operator; then $A$ is $\omega$-local in $\rcal$ if and only if $A \in \A(\rcal)$.

\item \label{it:aclosed}
  Let $A\in\qf^\omega$ be $\omega$-closable. Suppose that  
\begin{equation}\label{eq:weyldomaincond}
\forall g \in \dcal_\rbb^\omega(\rbb^2): \quad  \exp (i \phi(g)^-) \fpno \subset \domain{A^-}.
\end{equation}
  Then $A$ is $\omega$-local in $\rcal$ if and only if $A^-$ is affiliated with $\A(\rcal)$.

\item \label{it:aising}
 In the case $S=-1$, statement (\ref{it:aclosed}) is true even without the condition \eqref{eq:weyldomaincond}.

\end{enumerate}
\end{proposition}

\begin{proof}
We will prove the statement only for $\rcal=\wcal$ (the standard right wedge). For $\rcal=\wcal_x$ or $\rcal=\wcal_y'$,  it can then be obtained by applying Poincar\'e transformations, and for $\rcal=\ocal_{x,y}$  by considering intersections.
Also, (\ref{it:abounded}) is a special case of (\ref{it:aclosed}). 

For (\ref{it:aclosed}),
let $A$ be $\omega$-closable and $\omega$-local in $\wcal$. Let $g \in \dcal^\omega_\rbb(\wcal ')$. For $n \in \nbb_0$, we set
\begin{equation}
 B_n:=\sum_{k=0}^n \frac{i^k}{k!}(\phi(g)^-)^k,
\end{equation}
an operator defined at least on $\fpno$, along with its adjoint. Noting that powers of $\phi(g)$ leave $\fpno$ invariant, we can deduce from $\omega$-locality of $A$ by repeated application of Definition~\ref{definition:omegalocal} that
\begin{equation}\label{eq:phincommute}
     \bighscalar{B_n^\ast \psi}{ A \chi } = \bighscalar{ \psi}{ A B_n \chi }
= \bighscalar{ (A^-)^{\ast} \psi}{ B_n \chi }
\quad
\text{for all }
     \quad \psi,\chi \in \fpno,
\end{equation}
where the last equality uses $\psi \in \domain{(A^-)\st}$.
Both $\psi$ and $\chi$ are analytic vectors for $\phi(g)$. Therefore, as $n \to \infty$, we have $B_n\chi \to B\chi$ and $B_n\st\psi \to B\st\psi$, with $B:=\exp i \phi(g)^-$. Equation~\eqref{eq:phincommute} implies
\begin{equation}\label{eq:gfcommute}
     \bighscalar{B\st \psi}{ A^- \chi } = \bighscalar{ (A^-)^{\ast} \psi}{ B \chi }
\quad
\text{for all }
     \quad \psi,\chi \in \fpno.
\end{equation}
By the hypothesis \eqref{eq:weyldomaincond}, we have $B \chi \in \domain{A^-}$, which implies $\hscalar{(A^-)\st \psi}{B \chi} = \hscalar{\psi}{A^- B \chi}$. Since $B$ is bounded and $\psi$ can be chosen from a dense set in $\Hil$, we conclude that
\begin{equation}\label{eq:bafpncommute}
    B A^- \chi = A^- B \chi
\quad
\text{for all }
     \chi \in \fpno.
\end{equation}
If now more generally $\chi \in \domain{A^-}$, we can find a sequence $(\chi_j)$ in $\fpno$ such that $\chi_j \to \chi$ and $A^- \chi_j \to A^- \chi$ in $\Hil$. We compute from Eq.~\eqref{eq:bafpncommute}, using boundedness of $B$,
\begin{equation}\label{eq:balimit}
   B \chi_j \to B \chi \quad \text{and} \quad A^- B \chi_j = B A^-\chi_j \to B A^- \chi,
\end{equation}
so that, since $A^-$ is closed,
\begin{equation}\label{eq:badcommute}
\begin{aligned}
    B \chi \in & \domain{A^-} \quad  \text{and} \quad A^- B\chi = BA^-\chi
\quad \\
&\text{for all }
     \chi \in \domain{A^-},\;
      B = \exp(i \phi(g)^-),\; g \in \dcal^\omega_\rbb(\wcal').
\end{aligned}
\end{equation}
The same then holds if $B$ is a finite product of operators $\exp(i \phi(g)^-)$, a linear combination of those, or their strong limit (by a similar computation as in \eqref{eq:balimit}). Thus, by the double commutant theorem, \eqref{eq:badcommute} holds for all $B \in \{\exp(i \phi(g)^-) \,\big|\, g \in \dcal^\omega_\rbb(\wcal')  \}''=\A(\wcal)'$. But this means $A^- \affiliated \A(\wcal)$, as claimed.

For the converse, let $A^- \affiliated \A(\wcal)$ and let $g \in \dcal_\rbb(\wcal')$. For any $t \in \rbb$, we have $\exp i t \phi(g)^- \in \A(\wcal)'$. Affiliation of $A$ implies
\begin{equation}\label{eq:wlconverse}
   \forall \psi,\chi \in \fpno\; \forall t \in \rbb: \quad
    \hscalar{e^{-i t \phi(g)^-} \psi}{A^- \chi}
=   \hscalar{(A^-)^\ast \psi}{e^{i t \phi(g)^-} \chi}.
\end{equation}
Since $\psi,\chi$ are analytic vectors for $\phi(g)$, both sides of \eqref{eq:wlconverse} are real analytic in $t$. Computing their derivative at $t=0$, we find
\begin{equation}
   \forall \psi,\chi \in \fpno: \quad
    \hscalar{ \phi(g) \psi}{A \chi}
=   \hscalar{(A^-)\st \psi}{ \phi(g) \chi} =   \hscalar{ \psi}{ A \phi(g) \chi}.
\end{equation}
Since $g \in \dcal_\rbb(\wcal')$ was arbitrary, and since we can extend the relation to complex-valued $g$ by linearity, this means that $A$ is $\omega$-local in $\wcal$. This completes the proof of (\ref{it:aclosed}).

For (\ref{it:aising}), note that in the case $S=-1$, the operators $\phi(g)^-$ are actually bounded, and generate the algebra $\A(\wcal)'$ \cite{Lechner:2005}; we can restrict to $g \in \dcal_\rbb^\omega(\wcal')$ here by density. It is clear that $\phi(g) \fpno \subset \fpno\subset \domain{A^-}$, and using this instead of \eqref{eq:weyldomaincond}, a similar (in fact, simpler) computation as for (\ref{it:aclosed}) shows that $A^- \affiliated \A(\wcal)$.
\cmpqed\end{proof}

\subsection{Cyclicity of the vacuum, and relation to the local algebras}  \label{sec:reehschlieder}

We now ask for criteria which guarantee that a set of (local) quadratic forms is ``maximally large'', in the sense of generating all vectors in the Hilbert space, or all local observables in a certain sense.

We first investigate the Reeh-Schlieder property, i.e., the question whether the vacuum is cyclic for given subspaces $\qf_{\mathrm{d}} \subset \qf^\omega$ of quadratic forms. More specifically, we suppose that each $A \in \qf_{\mathrm{d}}$ is $\omega$-closable, hence $A^-\Omega$ is well defined; we ask whether $\{A^- \Omega: A \in \qf_{\mathrm{d}}\}$ is dense in $\Hil$.

We show that for cyclicity, it is sufficient to check density of the states at finite particle number over compact sets in rapidity space only. To that end, for $m \in \nbb_0$, we denote with $P_{m}$ the projector onto $\hcal_m\subset \hcal$ as before, and for $M\subset\nbb_0$ we write $P_{M} := \sum_{m\in M} P_{m}$.  Further, let  $P_{m,\rho}$ be the subprojection of $P_m$  onto functions supported in the ball of radius $\rho>0$, and $P_{M,\rho}$ accordingly.
\begin{lemma}\label{lemma:densitycrit}
 Let $M \subset \nbb_0$. Let $\qf_{\mathrm{d}} \subset \qf^\omega$ be a subspace with the following properties:
 \begin{enumerate}[(i)]
  \item \label{it:vacdom}
  Each  $A \in \qf_{\mathrm{d}}$ is $\omega$-closable.
  \item \label{it:timetrans} 
  For each $A \in \qf_{\mathrm{d}}$ there exists $\epsilon>0$ such that $A(x):=U(x) A U(x)^\ast \in \qf_{\mathrm{d}}$ whenever $|x|<\epsilon$.
  \item \label{it:protodense} 
  For each \emph{finite} subset $N \subset M$, and each $\rho>0$, the inclusion
  \begin{equation} 
\{  P_{N, \rho} A^- \Omega : A \in \qf_{\mathrm{d}}\} \subset  P_{N, \rho} \hcal \qquad \text{is dense.}
  \end{equation}
 \end{enumerate}
Then, $\{P_M A^- \Omega : A \in \qf_{\mathrm{d}}\}$ is dense in $P_M \hcal$.
\end{lemma}

\noindent
\emph{Remarks:} Condition (\ref{it:vacdom}) can be replaced with the weaker requirement that each $A$ extends to an operator with $\Omega$ in its domain. In applications in Secs.~\ref{sec:evenexamples}--\ref{sec:oddexamples}, $M$ will either be the set of even or of odd numbers, as we need to treat even and odd particle numbers separately.

\begin{proof}
 Let $\psi\in \hcal$ be orthogonal to $P_M A^- \Omega$ for all $A \in \qf_{\mathrm{d}}$; we need to show $P_M \psi = 0$. We apply a variant of the well-known Reeh-Schlieder argument \cite{ReehSchlieder:felder}. To that end, let $e$ be the unit vector in time direction, and consider for fixed $A \in \qf_{\mathrm{d}}$ the function
 \begin{equation}
    \rbb \ni t \mapsto \hscalar{\psi}{ P_M A(te)^- \Omega } = \hscalar{P_M \psi}{ U(te) A^-  \Omega }, 
 \end{equation}
 which is well-defined and continuous due to (\ref{it:vacdom}). It vanishes for $|t|<\epsilon$ due to (\ref{it:timetrans}). On the other hand, due to the spectrum condition for $U$, it is the boundary value of a function analytic in the upper half-plane, which must therefore vanish identically. Computing its Fourier transform in the sense of distributions, we see that
 \begin{equation}\label{eq:psisum}
    0 = \sum_{m \in M} \int \frac{d^m \thetav}{m!} \, \overline{ \psi_m(\thetav)} \, h(E(\thetav) )  \, \cme{m,0}{A}(\thetav)
    \quad \text{for all } A\in \qf_{\mathrm{d}}, \, h \in \scal(\rbb).
 \end{equation}
 In particular, for given $q>0$, we can choose $h$ to equal 1 on $[-q,q]$ and 0 outside $[-2q,2q]$. Since $E(\thetav) \geq m$, the sum \eqref{eq:psisum} is then finite ($m \leq 2q$), and the integration can be restricted to the compact region $E(\thetav) \leq 2q$. Due to (\ref{it:protodense}) with suitably chosen $\rho$, we then conclude that $\psi_m(\thetav)$ vanishes when $m\in M$, $m \leq 2q$, for (almost every) $\thetav$ in the support of $h(E(\cdotarg))$ -- that is, at least where $E(\thetav) \leq q$. Now letting $q \to \infty$, we see that $\psi_m(\thetav)=0$ for all $m\in M$ and almost all $\thetav$, i.e., $P_M\psi = 0$.
\cmpqed\end{proof}

From here, if the $A \in \qf_{\mathrm{d}}$ are affiliated with some algebra $\A(\ocal)$, we can deduce that the local algebra has the Reeh-Schlieder property. But more is true. To that end, consider the ``locally generated'' net of algebras,
\begin{equation}
 \rcal \mapsto \A_\mathrm{loc}(\rcal) := \bigvee_{\ocal \subset \rcal}\A(\ocal),
\end{equation}
where $\ocal$ runs over all double cones; $\A_\mathrm{loc}(\rcal)$ is defined for all open (bounded or unbounded) regions $\rcal$ in Minkowski space. Then $\A_\mathrm{loc}(\rcal) \subset \A(\rcal)$ holds for all regions $\rcal$, and $\A_\mathrm{loc}(\ocal) = \A(\ocal)$ for all double cones $\ocal$; but for wedges, one may question whether equality holds. We show this, as well as Haag duality for double cone regions and a version of weak additivity, based on a space of quadratic forms associated with one fixed double cone. (Similar results were obtained on the basis of the split property in \cite[Sec.~2]{Lechner:2008}; we derive them here from a sufficiently large set of local quadratic forms, mostly with standard techniques.)

\begin{theorem}\label{theorem:localgen}
 Let $\hat \ocal$ be a double cone. Suppose there exists a subspace $\hat\qcal \subset \qf^\omega$ such that:
 \begin{enumerate}[(i)]
  \item \label{it:closable} 
  Each $A \in \hat\qf$ is $\omega$-closable, and $A^-$ is affiliated with $\A(\hat\ocal)$.
  \item \label{it:shift} 
  For each $A \in \hat\qf$ there exists $\epsilon>0$ such that $U(x) A U(x)^\ast \in \hat\qf$ whenever $|x|<\epsilon$.
  \item  \label{it:findense}
  For each finite set $N \subset \nbb_0$, and each $\rho>0$, the inclusion $\{  P_{N, \rho} A^- \Omega : A \in \hat\qf \} \subset  P_{N, \rho}  \hcal $ is dense.
 \end{enumerate}
 Then we have:
 \begin{enumerate}[(a)]
  \item \label{it:reehschlieder} \emph{Reeh-Schlieder property:} $\Omega$ is cyclic and separating for $\A(\hat\ocal)$ and for $\A(\hat\ocal)'$. 
  \item \label{it:genwedge} \emph{Locally generated wedge algebras:} $\A_\mathrm{loc}(\wcal_x) = \A(\wcal_x)$, $\A_\mathrm{loc}(\wcal_y') = \A(\wcal_y')$ for all $x,y$. 
  \item \label{it:haagdual} \emph{Haag duality:} $\A(\ocal)' = \A_\mathrm{loc}(\ocal')$ for all double cones $\ocal$.
  \item \label{it:weakadd}  \emph{Weak additivity:} If $e\in\rbb^2$ is timelike or lightlike, then $\bigvee_{t\in\rbb} \A(\hat\ocal + te) = \boundedops$.
 \end{enumerate}

\end{theorem}

\begin{proof}
For (\ref{it:reehschlieder}): By Lemma~\ref{lemma:densitycrit} with $M = \nbb_0$,  $\hat\qf\Omega$ is dense in $\hcal$. Now write the closure of a fixed $A \in \hat\qf$ in polar decomposition, $A^- = V \int_0^\infty \lambda dP(\lambda)$; then for each $L>0$, we have $A_L := V \int_0^L  \lambda dP(\lambda) \in \A(\hat\ocal)$, and $A_L \Omega \to A \Omega$ as $L \to \infty$. Hence $\A(\hat\ocal)\Omega$ is dense in $\hcal$ as well. By covariance, $\Omega$ is therefore cyclic for each $\A(\hat\ocal+x)$, hence in particular for $\A(\hat\ocal)'\supset \A(\hat\ocal+x)$ (with suitable spacelike $x$). It follows immediately that $\Omega$ is separating. ---For (\ref{it:genwedge}), consider the subspace $\hat \hcal := \cup_{\ocal\subset\wcal} \A(\ocal)\Omega \subset \operatorname{dom} \Delta^{1/2}$, where $\Delta$ is the modular operator of the wedge algebra $\A(\wcal)$, and $\ocal$ runs over all double cones. Since $\hat\ocal + x \subset \wcal$ for suitable $x$, we know from  (\ref{it:reehschlieder}) that $\hat\hcal$ is dense in $\hcal$. Further, $\hat \hcal$ is  invariant under the modular group $\Delta^{it}$, as these coincide with the Lorentz boosts. Therefore $\hat\hcal$ is a core for $\Delta^{1/2}$, cf.~\cite[Ch.~II Prop.~1.7]{EngelNagel:semigroups}. 
This however implies that $\cup_{\ocal\subset\wcal} \A(\ocal) \subset \A(\wcal)$ is strong-operator dense \cite[Theorem~9.2.36]{KadRin:algebras2}, which proves the claim for the standard wedge $\wcal$. For other wedges it follows by covariance.---%
For (\ref{it:haagdual}), let $\ocal=\ocal_{x,y}$ and observe that 
\begin{equation}
\begin{split}
 \A(\ocal_{x,y})' = \A(\wcal_x') \vee \A(\wcal_y) = \A_\mathrm{loc}(\wcal_x') \vee \A_\mathrm{loc}(\wcal_y) 
 \ahponly{\\}= \A_\mathrm{loc}(\wcal_x'  \cup \wcal_y) = \A_\mathrm{loc}(\ocal_{x,y}'),
\end{split}
\end{equation}
where the first step uses Haag duality for wedges, the second step employs (\ref{it:genwedge}), and the third follows from the definition of $\A_\mathrm{loc}$.%
---
For (\ref{it:weakadd}), cf.~\cite[Appendix]{Kuckert:regions}: Let $A \in \A(\hat\ocal)$ and $B \in (\bigcup_{t\in\rbb} \A(\hat\ocal + te))'$. Then the function $f(t):=\hscalar{A\Omega}{U(-te) B \Omega} = \hscalar{B^\ast\Omega}{U(te) A^\ast \Omega}$ has a bounded analytic continuation to both the upper and the lower halfplanes, since $U$ satisfies the spectrum condition. Hence $f$ must be constant. Since $U(te)$ weakly converges to the projector onto $\Omega$ as $t \to \infty$, we find $\hscalar{A\Omega}{B \Omega} = \hscalar{A \Omega}{\Omega} \hscalar{\Omega}{B\Omega}$. Then (\ref{it:reehschlieder}) implies $B = \hscalar{\Omega}{B\Omega} \idop$, showing the statement.
\cmpqed\end{proof}

There remains the question whether the space $\hat\qcal$ generates the algebra $\A(\hat\ocal)$ in some sense. This cannot follow directly from the above: Namely, if we consider $\hat \ocal:=\ocal_r$ for some fixed $r>0$, and
\begin{equation}
  \hat\qcal := \bigcup_{0 < c < 1/2} \A(\ocal_{cr}) \subset \A(\ocal_{r/2}),
\end{equation}
then this $\hat\qcal$ fulfills all conditions of Theorem~\ref{theorem:localgen} in typical models, but it is certainly not dense in $\A(\ocal_r)$. 

We can however deduce a result on the level of \emph{nets} of algebras. Suppose that for \emph{every} double cone $\ocal$ we are given a space $\qcal(\ocal)$ of quadratic forms affiliated with $\A(\ocal)$, subject to certain consistency conditions (see below). We then define the algebras
\begin{equation}\label{eq:qnet}
\begin{aligned}
   \A_\qcal(\ocal) := \big\{ V,V^\ast,P(\lambda) :  V \textstyle\int \lambda \,dP(\lambda) \; \ahponly{&} \text{is the polar decomposition }
   \ahponly{ \\ &}
    \text{of $A^-$ for some } A \in \qcal(\ocal)  \big\}''
 \end{aligned}
\end{equation}
for double cones $\ocal$, and $\A_\qcal(\rcal):=\bigvee_{\ocal\subset\rcal} \A_\qcal(\ocal)$ for other regions $\rcal$. Then clearly $\A_\qcal(\ocal) \subset \A_\mathrm{loc}(\ocal) \subset \A(\ocal)$. We show that $\A$ can be obtained from the subnet $\A_\qcal$ by dualization:

\begin{theorem}\label{theorem:dualnet}
  Let $\ocal \mapsto \qcal(\ocal)$ be a map from double cones to subspaces of $\qcal^\omega$ such that conditions (i)--(iii) of Theorem~\ref{theorem:localgen} are fulfilled for every $\hat\qcal=\qcal(\hat\ocal)$, and in addition:
  \begin{enumerate}[(i)]
   \setcounter{enumi}{3}
   \item \label{it:iso} For any two double cones $\ocal_1 \subset \ocal_2$, it holds that $\qcal(\ocal_1)\subset\qcal(\ocal_2)$.
   \item \label{it:cov} \sloppy For any double cone $\ocal$ and Poincar\'e transformation $(x,\Lambda)$, it holds that $U(x,\Lambda) \qcal(\ocal) U(x,\Lambda)^\ast=\qcal(\Lambda\ocal+x)$.
  \end{enumerate}
  Then, $\A_\qcal$ as defined in \eqref{eq:qnet} is a local, isotonous, covariant net of von Neumann algebras; $\Omega$ is cyclic for $\A_\qcal(\rcal)$ if $\rcal$ is nonempty, and separating if $\rcal'$ is nonempty; and $\A$ is the dual net of $\A_\qcal$, in the sense that for every double cone $\ocal$,
  \begin{equation}
     \A_\qcal(\ocal')' = \A(\ocal).
  \end{equation}
\end{theorem}

\begin{proof}
As $\A_\qcal(\rcal)\subset\A(\rcal)$ for every $\rcal$, locality is automatic; isotony and covariance follow from (\ref{it:iso}) and (\ref{it:cov}), respectively. $\Omega$ is cyclic and separating by Theorem~\ref{theorem:localgen}(\ref{it:reehschlieder}). Also, thanks to covariance of $\A_\qcal$, one obtains with methods as in Theorem~\ref{theorem:localgen}(\ref{it:genwedge}) that $\A_\qcal(\wcal_x) = \A(\wcal_x)$, $\A_\qcal(\wcal_y') = \A(\wcal_y')$ for any $x,y$. Hence we have
\begin{multline}
  \A_\qcal(\ocal_{x,y}')' =  \A_\qcal(\wcal_x' \cup \wcal_{y} )' = (\A_\qcal(\wcal_x') \vee \A_\qcal (\wcal_{y}) )'
 \\
 = (\A(\wcal_x') \vee \A (\wcal_{y}) )'  
  = \A(\wcal_x) \cap \A(\wcal_y') = \A(\ocal_{x,y}) 
\end{multline}
for any double cone $\ocal_{x,y}$.
\cmpqed\end{proof}


\section{Examples of local operators: even case}\label{sec:evenexamples}

We now illustrate the above methods for constructing local observables in examples; specifically, we will in a moment specialize to the massive Ising model, defined by the scattering function $S=-1$.

In view of the results in Sec.~\ref{sec:criterion}, our strategy will be as follows: We define meromorphic functions $F_k$ that satisfy the conditions (FD1)--(FD6) for some $r>0$, guided by experience from the form factor program. By Theorem~\ref{theo:conditionFD}, the associated quadratic form
\begin{equation}\label{eq:Fseries}
  A = \sum_{m,n} \int \frac{d^m\theta d^n \eta}{m!n!} F_{m+n}(\thetav+i\zerov,\etav-i\zerov) z^{\dagger m}(\thetav) z^n(\etav)
\end{equation}
is then $\omega$-local in the double cone $\ocal_r$. Separately, we show using summability criteria (Proposition~\ref{proposition:summable1} or \ref{proposition:summable2}) that $A$ is also $\omega$-closable. Then $A^-$ is affiliated with $\A(\ocal_r)$ by Proposition~\ref{proposition:locality}. If we construct sufficiently large sets of such $A$, fulfilling additional constraints such as isotony, covariance, and a density condition in compact regions of rapidity space, then the results in Sec.~\ref{sec:reehschlieder} imply the Reeh-Schlieder property for the local algebras, Haag duality, and that the $\A(\ocal)$ are generated by our quadratic forms via duality.

For all scattering functions $S$ in our class, the overall theory is invariant under the $\mathbb{Z}_2$-symmetry that replaces the wedge-local field $\phi$ with $-\phi$. As a consequence, all local observables can be split into an even and an odd part, in which only the even- and odd-numbered $F_j$ contribute, respectively. Hence it suffices to consider even and odd observables separately. 

Specific to the Ising model is the fact that for \emph{even} observables, the recursion relations \ref{it:fdrecursion} simplify considerably, since the factor on the right-hand side vanishes, hence the relation does not link $F_k$ and $F_{k+2}$. We will consider this simpler case in the present section, and the case of odd observables in Sec.~\ref{sec:oddexamples}.

In the even case of the Ising model, we can hence choose only \emph{one} of the functions $F_{2k}$ to be analytic and nonzero. This may seems an uninteresting special case at first glance; yet it comprises physically important observables, such as the averaged energy density $T^{00}(g)$ (see, e.g., \cite{BostelmannCadamuroFewster:2013}) whose only nonvanishing coefficient function is 
\begin{equation}\label{eq:edensity}
F_2(\zetav) =-i \mu^2 \sinh \frac{\zeta_1 - \zeta_2}{2} \sinh^2 \frac{\zeta_1 + \zeta_2}{2} \tilde g(p(\zetav)),
\end{equation}
where $g \in \scal(\rbb^2)$ is nonnegative.\footnote{We use the Fourier transform with the convention $\tilde{g}(p) := \frac{1}{2\pi}\int d^2 x\; e^{i p \cdot x}g(x)$.} 
Other examples of this type haven been given by Buchholz and Summers \cite{BuchholzSummers:2007} by considering even polynomials of the field $\phi(f)$.

We aim at constructing a large enough set of observables so that the Reeh-Schlieder property is fulfilled for these. To that end, let $k\in\nbb_0$, let $g \in \dcal(\rbb^2)$, and let $P\in\Lambda_{2k}^\pm$ be a symmetric Laurent polynomial in $2k$ variables (see Appendix~\ref{app:poly} for notational conventions). We define a sequence of analytic functions
\begin{equation}\label{eq:f2k}
F_{j}^{[2k,P,g]}(\zetav) := \begin{cases} \tilde{g}(p(\zetav))  P(e^{\zetav}) M\even_{2k}(\zetav) \quad & \text{for $j=2k$},\\ 0 & \text{otherwise}, \end{cases}
\end{equation}
where
\begin{equation}
M\even_{2k}(\zetav) := \sum_{\sigma \in \mathfrak{S}_{2k}} \sign \sigma \prod_{j=1}^k \sinh  \frac{\zeta_{\sigma(2j -1)} - \zeta_{\sigma(2j)}}{2} ,
\end{equation}
and by convention, $M\even_0 := 1$.
We claim that these functions fulfill our locality conditions.

\begin{proposition}
   Let $k\in\nbb_0$, $P\in\Lambda^\pm_{2k}$, and $g \in \dcal(\ocal_r)$ be fixed, with some $r>0$. Then $F_j^{[2k,P,g]}$ enjoy the properties \ref{it:fdmero}--\ref{it:fdboundsimag} with respect to this $r$ and both
   \begin{enumerate}[(a)]
    \item $\omega(p)= \beta \log (1+p)$ with sufficiently large $\beta>0$ for given $P$ and $k$, and
    \item $\omega(p)=p^\alpha$ with any fixed $\alpha\in(0,1)$, independent of $P$ and $k$.
   \end{enumerate}
\end{proposition}

\begin{proof}
   We drop the superscript $[2k,P,g]$.
   Since $g$ has compact support, $\tilde g$ and hence $F_{2k}$ are entire analytic \ref{it:fdmero}. 
   Also, \ref{it:fdsymm} means antisymmetry in our case, which is fulfilled since $M\even_{2k}$ is antisymmetric by construction.  
   Further, one finds that $M\even_{2k}(\zeta_1,\dotsc,\zeta_{n-1},\zeta_n+2\pi i) = -M\even_{2k}(\zetav) = M\even_{2k}(\zeta_n,\zeta_1,\dotsc,\zeta_{n-1})$, and the other factors are $2\pi i$-periodic in each variable, yielding \ref{it:fdperiod}. 
   As already noted, \ref{it:fdrecursion} is fulfilled since the right-hand side vanishes, and since all factors in $F_{2k}$ are entire analytic.
   For the estimates \ref{it:fdboundsreal} and \ref{it:fdboundsimag}, we need to consider only $\omega(p)=\beta\log(1+p)$ for suitable $\beta$, since $p^\alpha$ for any $\alpha\in(0,1)$ grows faster at large $p$.
   We first note that with suitable constants $c_1,c_2,c_3>0$,
   \begin{equation}\label{eq:mp2kest}
      \big\lvert M\even_{2k}(\zetav) \big\rvert  \leq c_1 \prod_{j=1}^{2k} (\cosh \re \zeta_j)^{\tfrac{1}{2}},
\qquad      \big\lvert P(e^{\zetav}) \big\rvert  \leq c_2 \prod_{j=1}^{2k} (\cosh \re \zeta_j)^{c_3}.      
   \end{equation}
   For \ref{it:fdboundsreal}, it suffices to show that for $\thetav \in \rbb^m$, $\etav\in \rbb^{2k-m}$ with any $m\in\{0,\ldots,2k\}$, and for suitably large $\beta$, the functions 
   \begin{equation}\label{eq:efe}
   \begin{aligned}
(\thetav,\etav) &\mapsto E(\thetav) F_{2k}(\thetav,\etav+i\piv) E(\etav)^{1-\beta} \quad \text{and} 
\\  (\thetav,\etav) &\mapsto E(\thetav)^{1-\beta} F_{2k}(\thetav,\etav+i\piv) E(\etav)     
\end{aligned}   \end{equation}
  are bounded. Namely, given this, we know that $(\thetav,\etav) \mapsto F_{2k}(\thetav,\etav+i\piv) E(\etav)^{-\beta}$ and $(\thetav,\etav) \mapsto E(\thetav)^{-\beta}F_{2k}(\thetav,\etav+i\piv)$ are square integrable, implying by Cauchy-Schwarz that $\|F_{2k}(\cdotarg, \cdotarg + i \piv)\|^{\omega}_{m \times n}$ is finite, and likewise $\|F_{2k}(\cdotarg -i \piv, \cdotarg) \|^{\omega}_{m \times n}$ by a redefinition of $P$ and $g$.
  
  We show boundedness of the first function in \eqref{eq:efe}, the other is similar. In fact, due to \eqref{eq:mp2kest}, and since $g$ is of Schwartz class, we can find $c_4>0$ such that 
  \begin{equation}
   \big\lvert E(\thetav) F_{2k}(\thetav,\etav+i\piv) E(\etav)^{1-\beta} \big\rvert \leq c_4  \frac{ E(\thetav)^{1+k+2kc_3} E(\etav)^{1 + k + 2k c_3 -\beta} }{1 + | E(\thetav) - E(\etav) |^{1+k+2k c_3} } ,
  \end{equation}
  which is bounded if $\beta \geq 2+2k+4kc_3$, since the function $(x,y) \mapsto \frac{xy^{-1}}{1 + \lvert x -y \rvert}$ is bounded for $x,y \geq 1$.
 
  For \ref{it:fdboundsimag}, from the support properties of $g$ we obtain by standard methods that for any $\zetav \in \cbb^{n}$,
  \begin{equation}\label{eq:gtildebound}
      \lvert \tilde g (p(\zetav)) \rvert \leq
       \frac{\lVert g \rVert_1}{2\pi} \prod_{j=1}^{n} e^{\mu r \lvert \im \sinh \zeta_j \rvert}.
  \end{equation}
  Using \eqref{eq:mp2kest}, we thus obtain
  \begin{equation}
     \lvert F_{2k} (\zetav)  \rvert \leq c_5  \prod_{j=1}^{2k} (\cosh \re\zeta_j)^{\frac{1}{2}+c_3} e^{\mu r \lvert \im \sinh \zeta_j \rvert}
  \end{equation}
  with some $c_5>0$. Choosing $\beta \geq \frac{1}{2}+c_3$, this is the desired bound \ref{it:fdboundsimag}. 
\cmpqed\end{proof}

Having verified \ref{it:fdmero}--\ref{it:fdboundsimag}, we can now apply our results to show that we obtain closable local operators.

\begin{proposition}\label{proposition:evenop}
Let $F_{j}^{[2k,P,g]}$ be defined as in \eqref{eq:f2k}, with $k \in \nbb_0$, $P\in\Lambda^\pm_{2k}$, and $g \in\dcal (\mathcal{O})$ for some double cone $\ocal$; and let $A^{[2k,P,g]}$ be the associated quadratic form.
Then, $A^{[2k,P,g]}$ is $\omega$-local in $\ocal$ with $\omega (p) =  \beta\log(1+p)$ where $\beta>0$ is sufficiently large for $P$ and $k$, and with $\omega (p) =  p^\alpha$ for any $\alpha\in(0,1)$. 
Further, $A^{[2k,P,g]}$ is $\omega$-closable, and its closure is affiliated with $\A(\ocal)$. 
\end{proposition}

\begin{proof}
By Poincar\'e covariance, we can assume without loss of generality that $\ocal=\ocal_r$. (Note that translations act only by shifting the argument of $g$, whereas boosts also scale the arguments of the polynomial $P$ by a constant factor;  cf.~\cite[Sec.~3.3]{BostelmannCadamuro:expansion}.)

Now the property of $\omega$-locality is a consequence of \ref{it:fdmero}--\ref{it:fdboundsimag} by Theorem~\ref{theo:conditionFD}. Closability follows from Proposition~\ref{proposition:summable2}, where the sum is actually finite; and Proposition~\ref{proposition:locality} proves affiliation.
\cmpqed\end{proof}

We now show that we have constructed \emph{all} (even) local quantities in the sense of the Reeh-Schlieder property.
To that end, let us define for any double cone $\ocal$,
\begin{equation}\label{eq:qeven}
   \qcal\even (\ocal) := \operatorname{span} \big\{ A^{[2k,P,g]} : k \in \nbb_0, P \in\Lambda_{2k}^\pm, g \in \dcal(\ocal)  \big\} \subset \qcal^\omega
\end{equation}
where $\omega(p)=p^\alpha$ with $\alpha\in(0,1)$, fixed in the following.
\sloppy
By the above remark, this is a covariant definition in the sense that $U(x,\Lambda) \qf\even(\ocal) U(x,\Lambda)^\ast = \qf\even (\Lambda\ocal + x)$.  With $P\even$ the projector onto the even particle number space within $\hcal$, we prove:

\begin{proposition}\label{proposition:evenrs}
  For any double cone $\ocal$, the inclusion $\qf\even(\ocal)\Omega \subset P\even \hcal$ is dense.
\end{proposition}

\begin{proof}
 Note that $A^{[2k,P,g]}\Omega \in P_{2k} \hcal$ are mutually orthogonal for different $k$. Therefore, it only remains to check the hypotheses (i)--(iii) of Lemma~\ref{lemma:densitycrit} for the case $\qf_{\mathrm{d}} = \qf\even(\ocal)$ and $P_M = P_{2k}$. 
 Here closability (i) follows from Proposition~\ref{proposition:evenop}. Further, since $\supp g$ is compact in the open set $\ocal$, slight translations of $g$ are contained in $\dcal(\ocal)$ as well, showing (ii). For (iii), we need to verify for any large ball $B\subset \rbb^{2k}$ that there is no antisymmetric square-integrable function $\psi\neq 0$ of $2k$ variables such that
 \begin{equation}\label{eq:psivanisheven}
     \int_B d\thetav\; \overline{\psi(\thetav)} \tilde g(p(\thetav)) M\even_{2k}(\thetav)  P(e^{\thetav} )  = 0 \quad \text{for all $P\in\Lambda_{2k}^\pm$, $g\in\dcal(\ocal)$}. 
 \end{equation}
 But since $\overline{\psi(\cdotarg)} \tilde g(p(\cdotarg)) M\even_{2k}$ is symmetric, and $\tilde g(p(\cdotarg)) M\even_{2k}$ vanishes only on a null set (due to analyticity), this follows from the density of polynomials on compact sets.
\cmpqed\end{proof}

We postpone duality results to Sec.~\ref{sec:oddrs}.

\section{Examples of local operators: odd case} \label{sec:oddexamples}

We now consider observables in the Ising model where the ``odd'' coefficients $F_{2k+1}$ are nonzero. Due to the recursion relations \ref{it:fdrecursion}, which are nontrivial in this case, we are forced to choose an infinite sequence of nonvanishing $F_{2k+1}$, linked to each other by their residues. 

Observables of this type have been considered in \cite{SchroerTruong:1978,BergKarowskiWeisz:1979,CardyMussardo:descendant}, among others; they include the so-called order parameter, or basic field, of the Ising model. Our particular focus is on closability of these quadratic forms, or put differently, on the summability of the expansion series \eqref{eq:Fseries}.

Specifically, we choose the sequences of meromorphic functions,  $k \in \nbb_0$, 
\begin{equation}\label{eq:foddnew}
F_{2k+1}\uidx{1,P,g}(\zetav) := \frac{1}{(2\pi i)^k} \tilde g(p(\zetav)) P(e^{\zetav}) M\odd_{2k+1}(\zetav), \quad F_{2k}\uidx{1,P,g}(\zetav) = 0 . 
\end{equation}
Here $g \in \mathcal{D}(\mathbb{R}^2)$ is a test function; we will make further restrictions on its momentum space behaviour below. $P$ is (essentially) a Laurent polynomial in any number of variables such that $P(y,-y,\xv) = P(\xv)$; we formalize the class $\Lambda\inv^\pm$ of these polynomials in Appendix~\ref{app:poly}, but let us note here that typical examples are the odd power sums, $\psp{2s+1}(\xv)=\sum_j x_j^{2s+1}$. The meromorphic function
\begin{equation}
  M\odd_{2k+1}(\zetav) := \prod_{1 \leq i < j \leq 2k+1} \tanh\frac{\zeta_i-\zeta_j}{2}
\end{equation}
will be further explored in Appendix~\ref{app:modd}.

As in Sec.~\ref{sec:evenexamples}, we want to verify that these functions indeed define local observables, and sufficiently many. We first check the more elementary properties \ref{it:fdmero}--\ref{it:fdrecursion} and \ref{it:fdboundsimag} in Sec.~\ref{sec:odd146}. Then we turn to \ref{it:fdboundsreal} and summability in Sec.~\ref{sec:odd5}, which involves delicate operator norm estimates of singular integral operators. Finally we derive the Reeh-Schlieder property and duality results in Sec.~\ref{sec:oddrs}.

Throughout this section, we will take $\omega(p)=p^\alpha$ with some $\alpha\in(0,1)$. We also define the function spaces
\begin{equation}\label{eq:dalpha}
\begin{aligned}
  \dcal_\alpha(\ocal) := \Big\{ g \in \dcal(\ocal) :  \sup_p \Big\lvert \ahponly{&} \exp(c \gnorm{p}{\infty}^\alpha) \big(\tfrac{\partial}{\partial p^0}\big)^{k_0} \big(\tfrac{\partial}{\partial p^1}\big)^{k_1} \tilde g(p) \Big\rvert < \infty 
  \ahponly{ \\ & } \text{ for all } k_0,k_1 \in \nbb_0 , c>0 \Big\},
  \end{aligned}
\end{equation}
where $\gnorm{p}{\infty}=\max\{\lvert p^0 \rvert,\lvert p^1 \rvert\}$. These $\dcal_\alpha(\ocal)$ are slightly different from $\dcal^\omega(\ocal)$. The only statement we will need about them is that for every open $\ocal$ they contain a nonzero function; in fact, $\dcal_\alpha(\ocal)$ is dense in $\dcal(\ocal)$, cf.~\cite[Sec.~1]{Bjoerck:1965}.

\subsection{Elementary properties}\label{sec:odd146}

We briefly state the results for \ref{it:fdmero}--\ref{it:fdrecursion} and \ref{it:fdboundsimag}, which can be deduced from the properties of $P$ and $M\odd_{2k+1}$ as explained in Appendices \ref{app:poly} and \ref{app:modd}, respectively.
\begin{proposition}\label{proposition:odd14}
   Let $P\in\Lambda\inv^{\pm}$ and $g \in \dcal(\rbb^2)$. The functions $F_j^{[1,P,g]}$ enjoy properties \ref{it:fdmero}--\ref{it:fdrecursion}.
\end{proposition}

\begin{proof}
    \ref{it:fdmero} is clear, since all factors are analytic where $|\im \zeta_m - \im \zeta_n| < \pi$, avoiding the poles of the hyperbolic tangent. 
Also, all factors are symmetric in their variables, except $M\odd_{2k+1}$ which is totally antisymmetric, yielding \ref{it:fdsymm}. 
Condition \ref{it:fdperiod} for $F_{2k+1}$ simply means $2\pi i$-periodicity in each variable, which is easy to verify for all factors in \eqref{eq:foddnew}.
For \ref{it:fdrecursion}, it is crucial to note that 
\begin{equation}
\res_{\zeta_2 - \zeta_1 = i\pi} M\odd_{2k+1}(\zetav) = -2 M\odd_{2k-1}(\hat{\zetav}) \quad \text{with} \; \hat\zetav = (\zeta_3,\dotsc,\zeta_{2k+1});  
\end{equation}
see Proposition~\ref{proposition:mprop}(\ref{it:mresidue}) in Appendix~\ref{app:modd}. Further, one has $p(\zetav)=p(\hat\zetav)$ at the residue (where $e^{\zeta_1}=-e^{\zeta_2}$) and $P(e^{\zetav})=P(e^{\hat \zetav})$ since $P \in \Lambda\inv^{\pm}$, cf.\ Appendix~\ref{app:poly}. 
This gives exactly \ref{it:fdrecursion} in the case $S=-1$. 
\cmpqed\end{proof}

\begin{proposition}\label{proposition:odd6}
   Let $P\in\Lambda\inv^\pm$ and $g \in \dcal(\ocal_r)$ with some $r>0$. The functions $F_j^{[1,P,g]}$ enjoy property \ref{it:fdboundsimag} with respect to this $r$.
\end{proposition}

\begin{proof}
We estimate $|F_{2k+1}^{[1,P,g]}(\zetav)|$ directly from \eqref{eq:foddnew}, and in doing so we bound: 
$\tilde g(p(\zetav))$ as in \eqref{eq:gtildebound}, where the support properties of $g$ enter; the Laurent polynomial $P$ by Proposition~\ref{proposition:pbound} in Appendix \ref{app:poly} (with $J=\emptyset$); and the function $M\odd_{2k+1}(\zetav)$ by Proposition~\ref{proposition:mprop}(\ref{it:mcomplexbd}) in Appendix~\ref{app:modd}. Combining these, we arrive at
\begin{equation}\label{eq:ftkpo}
\lvert F_{2k+1}^{[1,P,g]} (\zetav)  \rvert  \leq c_1  \lVert g \rVert_1  E(\re \zetav)^{c_2} 
\operatorname{dist} (\im\zetav,\ical^{2k+1}_\pm)^{-k} \prod_{j=1}^{2k+1} e^{\mu r \lvert \im \sinh \zeta_j \rvert}
\end{equation}
for all $\zetav \in \ical^{2k+1}_\pm$, 
with some $c_1,c_2>0$ (which may depend on $k$). Choosing $c_3>0$ such that $p \leq c_3\exp \omega(p) = c_3 \exp p^\alpha$ for all $p>0$, we then have
\begin{equation}
 E(\thetav) \leq (2k+1) \prod_{j=1}^{2k+1} \cosh \theta_j \leq (2k+1) c_3^{2k+1} \prod_{j=1}^{2k+1} \exp \omega(\cosh \theta_j) 
\end{equation}
for all $\thetav \in \rbb^{2k+1}$,
so that \eqref{eq:ftkpo} is in agreement with \ref{it:fdboundsimag}.
\cmpqed\end{proof}

\subsection{Operator domain and summability}\label{sec:odd5}

The remaining part for establishing $F_{2k+1}^{[1,P,g]}$ as the coefficients of a local operator is as follows. 
Setting $f_{mn}(\thetav,\etav) := F_{m+n}^{[1,P,g]}(\thetav + i\zerov,\etav+i\piv-i\zerov)$, 
we need to find bounds for the norm $\onorm{  f_{mn} }{m \times n}$. 
This will, first of all, establish \ref{it:fdboundsreal}.
However, we also need these estimates in order to show the summability of the series \eqref{eq:Fseries} when applied to a certain class of vectors, in order to extend $A$ to a closed operator.

The individual terms of the series are singular integral operators due to the poles of the $F_{m+n}$ along the integration contour,
and we have to find operator norm estimates for these. We start with a lemma to that end. In it, for a set of integers $J = \{j_1,\dotsc,j_\ell\}$,
we denote mixed partial derivatives of a function $h$ as $\partial_J h (\thetav) = \partial/\partial \theta_{j_1} \cdots \partial/\partial\theta_{j_\ell} h(\thetav)$; and where the function has additional arguments denoted $\etav$, the derivatives will not act on these.

\begin{lemma} \label{lemma:kernelsingular}
 Let $m,n \in \nbb_0$ and $0 \leq \ell \leq \min(m,n)$. Let $h : \rbb^m \times \rbb^{n} \to \cbb$ be smooth, and let $L:\rbb \backslash\{0\} \to \cbb$ be a continuous function,
 bounded outside a neighborhood of zero and analytic inside that neighborhood, except for a possible first order pole at $0$.
 Then, the integral kernel on $\rbb^m \times \rbb^n$,
 \begin{equation}
    K(\thetav,\etav) := h(\thetav,\etav) \prod_{j=1}^\ell L(\theta_j - \eta_j \pm i0),
 \end{equation}
 fulfills the bound
 \begin{equation}
   \gnorm{K}{m \times n}
\leq
  c_L^{m+n} \max_{J\subset\{1,\dotsc,\ell\}} \sup_{\thetav,\etav} \Big( \big\lvert \partial_J h (\thetav,\etav) \big\rvert \; \prod_{i=1}^m \sqrt{1+\theta_i^2} \prod_{j=1}^n \sqrt{1+\eta_j^2} \Big)  
  \end{equation}
  with a constant $c_L >0$ that depends on $L$ but not on $m$, $n$, $\ell$, or $h$.
\end{lemma}

\begin{proof}
We reduce the statement to special cases in four steps (a)--(d).

   (a) It suffices to prove the statement for $m=n=\ell$. Namely, once known for that case with some $c_L$, we can write for $\psi\in\dcal(\rbb^m)$, $\varphi\in\dcal(\rbb^n)$,
   \begin{equation}\label{eq:cssplit}
    \begin{aligned}
     \Big\lvert \int & d\thetav \, d\etav  \overline{\psi( \thetav )} K(\thetav,\etav) \varphi( \etav )  \Big\rvert
     \leq
     \int   \prod_{i=\ell+1}^m \frac{d\theta_i}{(1+\theta_i^2)^{1/2}} \prod_{j=\ell+1}^n \frac{d\eta_j}{(1+\eta_j^2)^{1/2}}        
        \\ &\times  \prod_{i=\ell+1}^m (1+\theta_i^2)^{1/2} \prod_{j=\ell+1}^n (1+\eta_j^2)^{1/2}
        \ahponly{ \\ & \times}
        \Big\lvert \int
         \overline{\psi(\thetav)} \varphi(\etav)  h(\thetav,\etav) \Big( \prod_{j=1}^\ell L(\theta_j - \eta_j \pm i0) d\theta_j d\eta_j \Big)   \Big\rvert.
    \end{aligned}
   \end{equation}
    We now apply the known statement to the inner integral and the Cauchy-Schwarz inequality to the outer integral, which yields the desired result as long as we choose $c_L \geq \sqrt{\pi}$.

   (b) If the statement holds for some $L$, and $M$ is a bounded continuous function analytic near $0$, then it holds for $L+M$ in place of $L$ as well.
   To see this, write
    \begin{equation}
    \begin{aligned}
    K(\thetav,\etav) &:= h(\thetav,\etav) \prod_{j=1}^k (L+M)(\theta_j - \eta_j \pm i0)
     \\ 
     &=  \sum_{J \subset \{1,\dotsc,\ell\}} h(\thetav,\etav) \prod_{j\in J} M(\theta_j - \eta_j) \prod_{j\in J^c} L(\theta_j - \eta_j \pm i0)  .
   \end{aligned}
   \end{equation}
   To each summand, we can now apply the statement for $L$ with $\ell=|J^c|$, which yields the result for $L+M$ with $c_{L+M} := \sqrt{2} c_L \sqrt{\gnorm{M}{\infty}+1}$.

   (c) It suffices to consider $L(\zeta)=\zeta^{-1}$. For other cases, set $a:=\res_{\zeta=0}L$ and note that $M(\zeta):= L(\zeta)-a/\zeta$ is analytic near 0, so that (b) can be applied.

   (d) Now let $\ell=m=n$ and $L(\zeta)=\zeta^{-1}$. We denote finite difference quotients of $h$ for $\ell=1$ as
       \begin{equation}\label{eq:fdone}
          \delta_{\xi} h (\eta) := \frac{h(\eta+\xi,\eta) - h(\eta, \eta)}{\xi}, 
       \end{equation}
       continued to $\xi=0$ by its limit, and for general $\ell$ as
       \begin{equation}
	  \delta_{J,\xiv} h (\etav) := \delta_{j_1,\xi_{j_1}} \ldots \delta_{j_a,\xi_{j_a}} h(\etav), 
       \end{equation}
       where $J=\{j_1,\dotsc,j_a\}$, and where $\delta_{j,\xi}$ acts on the $j$th component of the arguments like in \eqref{eq:fdone}. 
       As a special case,  $\delta_{\emptyset,\xiv} h (\etav) :=  h(\etav, \etav)$.
       With this notation, one has
\begin{equation}\label{eq:hsum}
   h(\etav+\xiv,\etav) =  \sum_{J \subset \{1,\dotsc,\ell\}} \delta_{J,\xiv} h(\etav) \prod_{j\in J} \xi_j , 
\end{equation}
and hence
\begin{equation}
   K(\thetav,\etav) =  \sum_{J \subset \{1,\dotsc,\ell\}} \delta_{J,\thetav-\etav} h(\etav) \prod_{j\in J^c} (\theta_j-\eta_j\pm i0)^{-1}.
\end{equation}
Note here that $\delta_{J,\thetav-\etav} h$  depends on $\theta_j$ only if $j \in J$, and that  $(\theta_j-\eta_j\pm i0)^{-1}$ ($j \in J^c$) thus acts as a bounded operator with respect to that variable.
Splitting the integration variables as in \eqref{eq:cssplit}, we find that
\begin{equation}
   \gnorm{K}{m \times n} \leq \sum_J c^{\ell-|J|} \sup_{\thetav,\etav}\Big( \big\lvert \delta_{J,\thetav-\etav} h(\etav)\big\rvert \prod_{j=1}^\ell \sqrt{1+\theta_j^2}\sqrt{1+\eta_j^2}\Big)
\end{equation}
where $c>0$ is some constant. Since the finite difference quotients are majorized by the corresponding partial derivatives, and the sum has $2^\ell$ terms, this proves the statement with $c_L:=\sqrt{2(c+1)}$. 
\cmpqed\end{proof}

We are interested in particular in the following kernels which appear as building blocks of the $F_{2k+1}$.

\begin{lemma}\label{lemma:kernelcoth}
 Let $P \in \Lambda\inv^\pm$, let $g \in \dcal_\alpha(\ocal)$, and let $k \in \nbb_0$ be fixed in the following.
 Consider the kernels on $\rbb^m\times\rbb^n$, $m,n \geq k$,
\begin{equation}\label{eq:kga}
  K(\thetav,\etav) = \frac{P(e^{\thetav},e^{-\etav}) \, \tilde g\big(p(\thetav)- p(\etav)\big)}{  \exp(E(\etav)^\alpha)}
 \prod_{j=1}^k \coth\frac{\theta_j-\eta_j \pm i0}{2}
\end{equation}
For each $\epsilon > 0$, there exists $c>0$ such that
 \begin{equation}\label{stima}
  \forall m,n \geq k: \quad \gnorm{K}{m \times n} \leq c^{m+n+1} m^{\epsilon m} n^{\epsilon n}.
 \end{equation}
\end{lemma}

\begin{proof}
 To apply Lemma~\ref{lemma:kernelsingular}, we need to estimate for $J\subset\{1,\dotsc,k\}$ the function
 \begin{equation}\label{eq:hdef}
  h_J(\thetav,\etav) := \partial_J f(\thetav,\etav)  \exp(-E(\etav)^\alpha) 
  \prod_{i=1}^m (1+\theta_i^2)^{1/2} \prod_{j=1}^n (1+\eta_j^2)^{1/2}
\end{equation}
where
\begin{equation} 
 f(\thetav,\etav) := P(e^{\thetav},e^{-\etav}) \; \tilde g\big(p(\thetav)- p(\etav) \big).
\end{equation}
We can explicitly compute
\begin{equation}\label{eq:fjderiv}
 \partial_J f =  \sum_{I \subset J } \partial_I P(e^{\thetav},e^{-\etav}) \;
  \Big(\big({\prod_{j \in J \backslash I}} \frac{dp}{d\theta}(\theta_j) \cdot \nabla_p \big) \tilde g(p) \Big)\Big\vert_{p=p(\thetav)-p(\etav)}.
\end{equation}
 Since $g\in\dcal_\alpha(\ocal)$, all derivatives of $\tilde g(p)$ until the order $|J|-|I|\leq k$ are bounded by $c_1 \exp(-\gnorm{p}{\infty}^{\alpha}/2\mu)$ with some $c_1>0$ (see Eq.~\eqref{eq:dalpha}). Further, $\gnorm{dp / d\theta}{1} \leq 2 \gnorm{p(\theta)}{\infty}$ for real arguments, and $\gnorm{p(\theta_j)}{\infty} \leq \mu E(\thetav)$. We also have $\gnorm{p(\thetav) - p(\etav)}{\infty} \geq  \mu \lvert  E(\thetav) - E(\etav) \rvert$. This yields
 \begin{multline}
 \Big\lvert \, \Big(\big({\prod_{j \in J \backslash I}} \frac{dp}{d\theta}(\theta_j) \cdot \nabla_p \big) \tilde g(p) \Big)\Big\vert_{p=p(\thetav)-p(\etav)} \Big\rvert
 \ahponly{\\} \leq (2 \mu E(\thetav))^{|J|-|I|}  c_1 \exp( -|E(\thetav)-E(\etav)|^{\alpha}/2 )
 \\ \leq (1+2 \mu E(\thetav))^k  c_1 \exp( -\frac{1}{2}E(\thetav)^{\alpha} + \frac{1}{2}E(\etav)^{\alpha}).
 \end{multline}
 The derivative $\partial_I P$ in \eqref{eq:fjderiv} can be estimated by Proposition~\ref{proposition:pbound}.
 With the sum in \eqref{eq:fjderiv} containing at most $2^{|J|} \leq 2^k$ terms, we obtain
 \begin{equation}
  | \partial_J f | \leq c_2^{k+1}  E(\thetav)^{c_3} E(\etav)^{c_3}  \exp( -\frac{1}{2}E(\thetav)^{\alpha} + \frac{1}{2}E(\etav)^{\alpha})  
 \end{equation}
 with constants $c_2$, $c_3$. Given $\epsilon>0$, we further estimate $(1+\theta_j^2)^{1/2} \leq c_4 E(\thetav)^{\epsilon}$ and similarly for $\etav$. Thus we finally obtain in \eqref{eq:hdef},
 \begin{equation}    
 \begin{aligned}
   |h_J(\thetav,\etav)| &\leq c_5^{m+n+1} E(\thetav)^{c_3} E(\etav)^{c_3} E(\thetav)^{m\epsilon} E(\etav)^{n\epsilon} \exp( -\frac{1}{2}E(\thetav)^{\alpha} - \frac{1}{2} E(\etav)^{\alpha} )
  \\
  &\leq c_5^{m+n+1} \sup_{z \geq 0}\big( z^{(c_3+m\epsilon)/\alpha} e^{-z/2}\big)\sup_{z \geq 0}\big( z^{(c_3+n\epsilon)/\alpha}e^{-z/2}\big)
  \\
  &\leq c_6^{m+n+1} m^{m\epsilon/\alpha} n^{n\epsilon/\alpha}.
 \end{aligned}
 \end{equation}
 With $L(\zeta) = \coth (\zeta/2)$, the result now follows from Lemma~\ref{lemma:kernelsingular} after a redefinition of $\epsilon$.
\cmpqed\end{proof}

With the help of this result, and knowledge of the singularity structure of the functions $M\odd_{2k+1}$ as developed in Appendix~\ref{app:modd}, we can now estimate the $\onorm{\cdotarg}{m \times n}$-norms of the expansion coefficients $f_{mn}$ of our proposed local operators.

\begin{proposition}\label{proposition:oddomeganorm}
 Let $f_{mn}(\thetav,\etav):=F_{m+n}^{[1,P,g]}(\thetav + i \zerov, \etav +i\piv - i \zerov)$, 
 with $g \in \dcal_\alpha(\ocal)$ for some bounded $\ocal$. 
 For fixed $\epsilon>0$ and $n \in \nbb_0$, there is a constant $c>0$ such that for all $m \in\nbb_0$,
\begin{equation}
   \onorm{f_{mn}}{m \times n} + \onorm{f_{nm}}{n \times m} \leq c^{m+1} m^{\epsilon m}.
\end{equation}
The same holds for $\hat f_{mn}(\thetav,\etav):=F_{m+n}^{[1,P,g]}(\thetav - i \piv + i \zerov, \etav - i \zerov)$.
In particular, the functions $F_j\uidx{1,P,g}$ fulfill condition \ref{it:fdboundsreal}.
\end{proposition}

\begin{proof}
 Representing $M\odd_{m+n}(\thetav+i\zerov,\etav +i\piv - i \zerov)$ as in Lemma~\ref{lemma:mlrsum}, we obtain with the notation introduced there,
\begin{equation}\label{eq:fmnothersum}
\begin{aligned}
   f_{mn}&(\thetav,\etav) =  \frac{\tilde g\big(p(\thetav)-p(\etav)\big)P(e^{\thetav},e^{-\etav})}{(2\pi i)^{m+n}}
    \\ & \times \ahponly{\mspace{-40mu}}
  \sum_{\substack{ k;(\ell_1,r_1),\dotsc,(\ell_k,r_k) \\ 1 \leq \ell_i \leq m < r_i \leq m+n}}
  \Big(\prod_{j=1}^k \coth \frac{\theta_{\ell_j}-\eta_{r_j-m}+i0}{2}\Big)
  (-1)^{s_{\ell r}} M\odd_{m-k}  (\hat \thetav) M\odd_{n-k} (\hat \etav).
  \end{aligned}
\end{equation}
Here $M\odd_{m-k}$, $M\odd_{n-k}$ are bounded by 1 (Proposition~\ref{proposition:mprop}(\ref{it:mrealbd})), hence they act as multiplication operators with norm $\leq 1$ with respect to the variables $\hat\thetav,\hat\etav$. Applying Lemma~\ref{lemma:kernelcoth} to each term of the sum \eqref{eq:fmnothersum}, knowing that the number of terms grows like $2^m$ at fixed $n$ (see Lemma~\ref{lemma:mlrsum}), then yields
\begin{equation}
   \gnorm{f_{mn}(\thetav,\etav)\exp(-E(\etav)^\alpha)}{m \times n} \leq c^{m+1}  m^{\epsilon m}
\end{equation}
where $c$ is a constant depending on $\epsilon,\alpha,g,n$ but independent of $m$.
Exchanging $m$ with $n$, and $\thetav$ with $\etav$, one obtains a similar result for $\gnorm{\exp(-E(\thetav)^\alpha)f_{mn}(\thetav,\etav)}{m \times n}$, and likewise for $f_{nm}$, which yields the result at $\zetav=(\thetav+i \zerov,\etav+i\piv-i\zerov)$ after a redefinition of constants. 
The computation at $\zetav=(\thetav-i\piv+i \zerov,\etav-i\zerov)$ is analogous, as $M\odd_{m+n}$ depends only on the differences of its variables. 
\cmpqed\end{proof}

With Proposition \ref{proposition:oddomeganorm}, we have shown that the $F\uidx{1,P,g}$ fulfill all conditions (FD). Therefore they yield $\omega$-local quadratic forms $A\uidx{1,P,g}$ via the series \eqref{eq:Fseries}. But our estimates suffice even for affiliation with the local algebras.

\begin{theorem}\label{theorem:oddoperator}
  Let $\omega(p)=p^\alpha$ with some $\alpha\in(0,1)$; let $P \in \Lambda\inv^\pm$, and $g \in \dcal_\alpha(\ocal)$ with some double cone $\ocal$; and let $F_j\uidx{1,P,g}$ be defined as in \eqref{eq:foddnew}. The associated quadratic form $A\uidx{1,P,g}\in\qf^\omega$ is $\omega$-closable, and its closure is affiliated with $\A(\ocal)$.
\end{theorem}

\begin{proof}
Again, by covariance, it suffices to consider $\ocal=\ocal_r$.
We saw through Propositions~\ref{proposition:odd14}, \ref{proposition:odd6} and \ref{proposition:oddomeganorm}, and via Theorem~\ref{theo:conditionFD}, that $A\uidx{1,P,g}$ is $\omega$-local in $\ocal_r$. Due to Proposition~\ref{proposition:oddomeganorm}, we have at fixed $n$,
\begin{equation}\label{eq:sumtoconverge}
   \sum_{m=0}^\infty \frac{2^{m/2}}{\sqrt{m!}} \Big( \onorm{ f_{mn} }{m \times n} + \onorm{ f_{nm}}{n \times m} \Big) 
  \leq  \sum_{m=0}^\infty c^{m+1} 2^{m/2} \frac{m^{\epsilon m}}{\sqrt{m!}}.
\end{equation}
By Stirling's approximation, the leading term in $m$ behaves like
\begin{equation}
  \frac{m^{\epsilon m }}{\sqrt{m!}} \sim \frac{m^{(\epsilon-1/2)m}e^{m/2}}{(2\pi m)^{1/4}} ,
\end{equation}
and choosing $\epsilon < 1/2$, the series then converges by the quotient criterion.
Therefore, Proposition~\ref{proposition:summable1} shows that $A\uidx{1,P,g}$ is $\omega$-closable, and Proposition~\ref{proposition:locality}(\ref{it:aising}) shows that its closure is affiliated with $\A(\ocal_r)$.
\cmpqed\end{proof}

\subsection{Reeh-Schlieder property}\label{sec:oddrs}

As for the even case in Sec.~\ref{sec:evenexamples}, we now ask whether we have found ``all'' local quantities. 
In analogy to \eqref{eq:qeven}, we define for any double cone $\ocal$, 
\begin{equation}
 \qf\odd(\ocal):= \operatorname{span} \big\{ A^{[1,P,g]} : P \in \Lambda\inv^\pm, g \in \dcal_\alpha(\ocal) \big\} \subset \qf^\omega.
\end{equation}
Its elements are $\omega$-closable due to Theorem~\ref{theorem:oddoperator}.
We then obtain the Reeh-Schlieder property for our quadratic forms in the following sense, with $P\odd$ the projector onto the odd particle number space in $\hcal$:

\begin{proposition}\label{proposition:oddrs}
  For any fixed $\ocal$, the inclusion $\qf\odd(\ocal)\Omega \subset P\odd \hcal$ is dense.
\end{proposition}

\begin{proof}
 As in the even case, small space-time translations leave each $A^{[1,P,g]} \in \qf\odd(\ocal)$ localized in $\ocal$. To apply Lemma~\ref{lemma:densitycrit}, it is then sufficient to show that the inclusion
 \begin{equation}\label{eq:psivanishodd}
   \operatorname{span} \Big\{ \bigoplus_{k=1}^{n}  \big(\thetav \mapsto \tilde g(p(\thetav)) M\odd_{2k+1}(\thetav) P(e^{\thetav} ) \big) : P \in \Lambda\inv^{\pm}\Big\} \subset \bigoplus_{k=1}^n L_{-}^2(\bcal_\rho^{2k+1})  
 \end{equation}
 is dense for any fixed $\rho>0$, $n\in\nbb$, and $g \in \dcal_\alpha(\ocal) \backslash \{0\}$, where $\bcal_\rho^{2k+1}\subset\rbb^{2k+1}$ is the ball of radius $\rho$, and $ L_{-}^2$ denotes the antisymmetric part of the $L^2$ space. Hence let $\psi_{2k+1} \in L_{-}^2(\bcal_\rho^{2k+1})$ ($k=0,\ldots,n$) be such that
 \begin{equation}
    0=\sum_{k=0}^{n} \int_{|\thetav| \leq \rho} d^{2k+1}\theta \; \overline{\psi_{2k+1}(\thetav)} \tilde g(p(\thetav)) M\odd_{2k+1}(\thetav) \, P( e^{\thetav} ) 
 \end{equation}
 for all $P \in \Lambda\inv^{\pm}$ (and some fixed $g$); we want to show that the $\psi_{2k+1}$ all vanish. By Proposition~\ref{prop:idense}, we know that
\begin{equation}
    0=\sum_{k=0}^{n} \int_{|\thetav| \leq \rho} d^{2k+1}\theta \; \overline{\psi_{2k+1}(\thetav)} \tilde g(p(\thetav)) M\odd_{2k+1}(\thetav) \, f_{2k+1}(\thetav)
 \end{equation}
 for any choice of continuous symmetric functions $f_{2k+1}$, since they can be uniformly approximated by $P( e^{\thetav} )$.
 Hence $\overline{\psi_{2k+1}(\thetav)} \tilde g(p(\thetav)) M\odd_{2k+1}(\thetav)$, which is symmetric, must vanish  a.e.; and since $\tilde g(p(\thetav)) M\odd_{2k+1}(\thetav)$ can vanish only on sets of measure zero due to analyticity, we conclude $\psi_{2k+1}=0$ a.e.\ for $k=0,\ldots,n$.
\cmpqed\end{proof}

Now we set $\qf(\ocal) := \qf\even(\ocal)+\qf\odd(\ocal)$; these $\qf(\ocal)$ are isotonous and covariant, they fulfill conditions (\ref{it:closable}) and (\ref{it:shift}) of Theorem~\ref{theorem:localgen} as before, and combining Propositions~\ref{proposition:evenrs} and \ref{proposition:oddrs}, we know that $\qf(\ocal)\Omega\subset\Hil$ is dense for every double cone $\ocal$. Let us consider the local net $\A_\qf$ generated by the spectral data of the operators in $\qf(\ocal)$, see \eqref{eq:qnet}. We can now apply Theorem~\ref{theorem:dualnet} and conclude: 
\begin{corollary}\label{cor:isingdual}
   For every double cone $\ocal$, the vector $\Omega$ is cyclic and separating for $\A_\qf(\ocal)$, and we have $\A_\qf(\ocal')' = \A(\ocal)$.  
\end{corollary}

Moreover, via Theorem~\ref{theorem:localgen}, we obtain the Reeh-Schlieder property and Haag duality for the local net $\A$ in the massive Ising model; this provides an independent new proof for earlier known results \cite{Lechner:2005,Lechner:2008}.


\section{Discussion of results} \label{sec:outlook}

In this paper we have explicitly constructed a set of local observables in the massive Ising model, as an example for an integrable quantum field theory. To that end, we defined sequences of meromorphic functions $F_k$ which are, essentially, solutions of the well-known form factor equations (more precisely, conditions (FD1)--(FD6) in Theorem~\ref{theo:conditionFD}). Via the series \eqref{eq:Fseries}, these $F_k$ define local operators, for which the main technical point is to control convergence of the series in a suitable sense. Proposition~\ref{proposition:summable1} gives a sufficient criterion in this respect, which we can indeed verify in relevant examples (Sec.~\ref{sec:odd5}). In fact, we have found sufficiently many examples to generate \emph{all} local observables in a well-defined sense (Corollary~\ref{cor:isingdual}).

This indicates that our approach can overcome the difficulties inherent in the convergence of $n$-point functions in the form factor program, and give mathematical meaning to local fields  in a more general sense, i.e., as closed operators affiliated with local von Neumann algebras.

Let us comment on some individual aspects of our results.

 \subsection{Operator content of the Ising model}
 
In the Ising model, we have shown that the sets of fields $\qf(\ocal)$ that we constructed, separated into even and odd parts, have the Reeh-Schlieder property, and that they generate the local algebras $\A(\ocal)$ by duality (see Sec.~\ref{sec:oddrs}, in particular Corollary~\ref{cor:isingdual}). In this sense, we can claim that we have constructed the full operator content of the Ising model. In particular, this provides an alternative proof that the $\A(\ocal)$ are nontrivial for any open $\ocal$, which was already shown in \cite{Lechner:2005}.

Specific elements of our class of observables include the order parameter $A\uidx{1,1,g}$ \cite{SchroerTruong:1978} and the energy density $A\uidx{2,P,g}$ with $P(x_1,x_2)= -\frac{i}{8} \mu^2 (x_1 x_2 + x_1^{-1} x_2^{-1} + 2)$, which plays an important role, e.g., in the study of quantum energy inequalities in integrable models \cite{BostelmannCadamuroFewster:2013,BostelmannCadamuro:2016}.

All operators that we have constructed are, in principle, pointlike fields smeared with test functions $g$ in space and time; their Fourier transform $\tilde g$ appears in the rapidity-space expansion coefficients $F_k$. We will elaborate more on their functional analytic aspects in Sec.~\ref{sec:pointlike}. Let us remark here that it would be sufficient to use averaging only in time; our results would be the same, but Poincar\'e covariance of the local observables would be less manifest.

The Laurent polynomial $P$ in the coefficients $F_k$ serves to enumerate the field content. For the Reeh-Schlieder property, usual polynomials $P\in\Lambda\inv$ in the odd case and $P\in\Lambda_{2k}$ in the even case would suffice, and indeed a subalgebra of $\Lambda\inv$ would already have the relevant density property (cf.\ the proof of Proposition~\ref{prop:idense}). However, the generalization to Laurent polynomials is important for applications: it allows us to include also derivatives of our fields, which act on the coefficients by multiplication with $p^j(\zetav)$, and related quantities such as the averaged energy density. 

It is interesting to note (cf.~the end of Sec.~2 in \cite{CardyMussardo:descendant}) that we obtain, among others, local observables $A$ that do not couple the vacuum with states of low particle number; that is, for some $n \in \nbb$, one has $A\Omega \perp \hcal_m$ for all $m < n$, but $A\Omega \not\perp \hcal_{n}$. In that respect, these $A$ are analogous to $n$-th Wick powers of a free field. Indeed, if $n$ is even, then every $A\uidx{n,P,g}$ has this property, and for $n$ odd, one can construct such operators $A\uidx{1,P,g}$ by including a factor of $J_{n}$ in the polynomial $P$, as discussed in Appendix~\ref{app:poly} (see Lemma~\ref{lemma:jprop}).

 \subsection{Pointlike fields in integrable models}\label{sec:pointlike}

 As mentioned, the observables we constructed in the Ising model have the structure of local averages of pointlike quantum fields. Formally replacing the averaging function $g$ with a delta function, and hence its Fourier transform with a constant, reproduces the well-known expressions from the form factor program. This ansatz can likely be carried over to other integrable models (see Sec.~\ref{sec:othermodels}). In this respect, the structure of our local fields is in line with expectations. However, our mathematical interpretation of these fields is very different from the usual approach: We construct them as (unbounded) operators, but we do not show, or require, that they exist as operator-valued distributions on a common invariant domain, neither in the axiomatic setting by Wightman \cite{StrWig:PCT} nor -- more aligned with our choice of test functions -- in its generalization by Jaffe \cite{Jaffe:1967}. 
 
 In particular, the closed extensions of our averaged fields $A\in\qf(\ocal)$ have a common dense core $\fpno \ni \Omega$, which is however not invariant under their action. Consequently, we do not make any statement about products of the field operators or about their $n$-point functions, beyond the 2-point function which exists since $\Omega \in \dom A^{-\ast}$. Instead, we can show -- at least in the Ising model -- that the closures of the field operators are affiliated with the (abstractly defined) local von Neumann algebras $\A(\ocal)$, and indeed that they generate the algebras $\A(\ocal)$ by duality.

 Note that, particularly for the Ising model, our claim is \emph{not} that $n$-point functions of local fields do not exist.  In fact, there are alternative constructions of (likely) the same model in a  Euclidean setting, where the Osterwalder-Schrader axioms can be verified \cite{PalmerTracy:ising}. Therefore one would expect that the expressions from the form factor program  actually do yield Wightman fields in the usual sense, fulfilling polynomial $H$-bounds, and hence field products should exist, even when using Schwartz-class test functions. But in a more general setting, the Wightman axioms might be too strict; and with our methods, we can interpret the fields meaningfully as local objects without the need of controlling the singular nature of operator products. In this sense, our results demonstrate that the $n$-point functions are \emph{not conceptually necessary}.
 
 An interesting question arises for the products of operators localized at spacelike distances. Namely, let $\ocal_1$ and $\ocal_2$ be two spacelike separated regions, and $A_1 \in \qf(\ocal_1)$, $A_2 \in \qf(\ocal_2)$. Even if $A_1^-$ and $A_2^-$ do not \emph{a priori} have a common invariant domain, they are affiliated with the commuting von Neumann algebras $\A(\ocal_1)$, $\A(\ocal_2)$, which means that $|A_1^-|$ and $|A_2^-|$ spectrally commute. Therefore, the product $A_1^-A_2^-=A_2^-A_1^-$ can be defined on a  suitably chosen domain. Thus there is hope to establish an operator product expansion with methods as in \cite{Bostelmann:products}, though the technical situation described there is somewhat different.

 \subsection{Other integrable models}\label{sec:othermodels}

 While we have carried out our full construction only in the massive Ising model, there is reason to believe that similar methods can be applied in other models as well. As far as a single species of massive scalar particles is concerned, the expansion \eqref{eq:expansionrecall} and the characterization of locality in Theorem~\ref{theo:conditionFD} apply independent of the scattering function $S$, and so do the criteria developed in Sec.~\ref{sec:criterion}. Candidates for local observables (i.e., form factors) are known in some of these models, most notably for the sinh-Gordon model \cite{FringMussardoSimonetti:1993}. Hence our methods should be applicable to the sinh-Gordon case in principle. Care is needed, however, since the form factors $F_k$ there have a more intricate structure, complicating the estimates at large $k$. Also, the extra condition \eqref{eq:weyldomaincond} in Proposition~\ref{proposition:locality} will need to be established outside the case $S=-1$.

 The situation is similar in models with a richer particle spectrum, such as the $O(N)$ nonlinear sigma models. Here form factors have been computed \cite{BabujianFoersterKarowski:2013}, and progress has been made towards the construction of the local algebras via wedge-local fields \cite{AL2017,Alazzawi:2015}. The expansion \eqref{eq:expansionrecall}, Theorem~\ref{theo:conditionFD}, and the criteria in Sec.~\ref{sec:criterion} have not yet been established for this case, but would be expected to generalize quite directly, using matrix-valued coefficient functions $F_k$. A challenge, of course, are the ever more complicated estimates on higher-order integral kernels $F_k$.
 
 A quite different problem arises in models with bound states, i.e., where the scattering function $S$ has poles in the physical strip $0 < \im \zeta < \pi$, such as the Bullough-Dodd, $Z(N)$-Ising, and sine-Gordon models. Here the form factor equations need to be modified, but solutions to them are known (see \cite{FringMussardoSimonetti:bulloughdodd,BabujianKarowski:zn,BabujianFringKarowskiZapletal:1993}, among others). However, on the side of the operator algebraic approach, the wedge-local fields can no longer have the simple form \eqref{eq:wlfield}. Work towards a construction of wedge-local fields and of local algebras $\A(\ocal)$ in this case has recently been carried out by one of the authors together with Y.~Tanimoto \cite{CT2015,CT2017,CTsine}. This gives hope that our present methods, with a suitably modified version of Theorem~\ref{theo:conditionFD}, can be applied to models with bound states as well.
 
 In particular, a generalization of our results might imply nontriviality of the local algebras $\A(\ocal)$ in cases where other methods have so far been unable to resolve this question: Our construction does not rely on the modular nuclearity condition or the split property for wedge algebras, rather it directly shows the existence of closable local operators and hence of their polar data.


\appendix
\normalsize

\section{Symmetric Laurent functions}\label{app:poly}

We discuss here a certain class of Laurent polynomials which are relevant in our constructions of operators in the Ising model, but more generally for ``descendant fields'' in integrable models of quantum field theory; see, e.g., \cite{CardyMussardo:descendant, Christe:descendant, FringMussardoSimonetti:1993}.

To that end, for $n \in \nbb$, we denote $\Lambda_n = \cbb\lbrack x_1,\dotsc,x_n\rbrack^{ \mathfrak{S}_{n}}$ the algebra of symmetric polynomials in $n$ variables, and $\Lambda_n^\pm = \cbb\lbrack x_1^{\pm 1},\dotsc,x_n^{\pm 1} \rbrack^{ \mathfrak{S}_{n}}$ the algebra of symmetric Laurent polynomials (i.e., polynomials which can contain negative powers of the $x_i$). 

For our purposes in particular in Sec.~\ref{sec:oddexamples}, we need a notion of Laurent polynomials ``independent of the number of variables''. Let $\Lambda$ be the algebra of symmetric functions (see, e.g., \cite{Macdonald:symmetric}), and $\varphi_n:\Lambda \to \Lambda_n$ the homomorphism that reduces a symmetric function to a polynomial in $n$ dimensions, i.e., $\varphi_n P (\xv) = P(x_1,\dotsc,x_n,0,0,\dotsc)$. Following \cite{SerVes:jacklaurent}, we define the \emph{algebra of symmetric Laurent functions} as $\Lambda^\pm=\Lambda \otimes \bar\Lambda$, where $\bar \Lambda$ is a copy of $\Lambda$ but read with respect to the ``inverse variables'' $x_i^{-1}$. More formally, we set $\varphi_n^\pm:\Lambda^\pm \to \Lambda_n^\pm$, $\varphi_n^\pm(P \otimes Q)(\xv) = (\varphi_n P)(x_1,\dotsc,x_n) \cdot  (\varphi_n Q)(x_1^{-1},\dotsc,x_n^{-1})$; this is compatible with $\varphi_n$ with respect to the natural inclusions $\Lambda\subset\Lambda^\pm$, $\Lambda_n\subset\Lambda_n^\pm$. The ring $\Lambda^\pm$ is freely generated by the power sum functions $\psp{k} = \sum_j x_j^k$, $k \in \zbb\backslash \{0\}$.  In the following, we will often write $P(\xv)$ rather than $\varphi_n^\pm P(\xv)$ for $\xv \in \rbb^n$, where no confusion can arise. 

For our purposes, we are particularly interested in functions with the property 
\begin{equation}\label{eq:genreduce}
   P(y,-y, \xv) = P(\xv)\quad \text{for all $n \in \nbb_0$ and $\xv \in \rbb^n$}.
\end{equation}
More formally, for  $y \in \rbb_+$, let $\alpha_y:\Lambda\to\Lambda$ be the homomorphism that substitutes $x_1 \to y$, $x_2 \to -y$, $x_{j+2} \to x_j$, and set $\alpha_y^\pm := \alpha_y \otimes \alpha_{1/y}$. Let $\Lambda\inv\subset \Lambda$ (respectively, $\Lambda\inv^\pm\subset \Lambda^\pm$) be the subalgebra of invariants under all $\alpha_y$ (respectively, $\alpha_y^\pm$). We are interested in characterizing these subalgebras.

\begin{proposition}\label{proposition:invargen}
   $\Lambda\inv$ and $\Lambda\inv^\pm$ are generated by the odd power sums $\psp{2k+1}$, $k \in \nbb_0$ and $k \in \zbb$, respectively. 
\end{proposition}

\begin{proof}
 It is clear that $\psp{2k+1}\in\Lambda\inv^\pm$. Now a generic element $P \in \Lambda^\pm$ is of the form
 \begin{equation}
     P = Q(\psp1,\psp{-1},\dotsc,\psp{2k+1},\dotsc,\psp{2},\psp{-2},\dotsc,\psp{2k},\dotsc)
 \end{equation}
 with some polynomial $Q$. Since $\alpha_y^\pm(\psp{2k}) = \psp{2k} + 2 y^{2k}$, we have for any $n \in \nbb$ and $y \in \rbb_+$,
 \begin{equation}
      (\alpha_y^\pm)^n P = Q(\dotsc,\psp{2k+1},\dotsc,\psp{2k} + 2 n y^{2k},\dotsc).
 \end{equation}
 If now $P \in \Lambda\inv^\pm$, then this expression is constant in $n$, even if we extend the r.h.s.\ to $n \in \rbb$ as a polynomial. Taking derivatives by $n$, this means
 \begin{equation}
    0 =  \sum_{k \neq 0} \frac{\partial Q}{\partial \psp{2k}} y^{2k}.
 \end{equation}
 This (finite) sum must vanish at every order in $y$; hence $Q$ is independent of all $\psp{2k}$, $k \neq 0$. This shows the statement for $\Lambda\inv^\pm$; the one for $\Lambda\inv$ is analogous.
\cmpqed\end{proof}

Hence we have a simple characterization of the invariant subalgebras. Other generators have been constructed in \cite{CardyMussardo:descendant, Christe:descendant}; we include them here for completeness: Let $\sigma_k = \sum_{i_1 < \dotsb < i_k} x_{i_1}\dotsm x_{i_k} \in \Lambda$ be the elementary symmetric polynomials, $k \in \nbb$. 
For $s \in \nbb_0$, set
\begin{equation}\label{eq:idef}
   I\lidx{2s+1} := \det \begin{pmatrix}                              
                             \sigma\lidx{2} & 1 & 0 & \dots & 0
                             \\
                             \sigma\lidx{4} & \sigma\lidx{2} & 1 & \dots & 0
                             \\
                             \vdots & & \ddots & \ddots & \vdots \\
                             \sigma\lidx{2s} & \sigma\lidx{2s-2} &  \ldots & \sigma\lidx{2} & 1
                             \\
                             \sigma\lidx{2s+1} & \sigma\lidx{2s-1} &  \ldots & \sigma\lidx{3} & \sigma\lidx{1}
                            \end{pmatrix} \in \Lambda.
\end{equation}
We also set $I\lidx{-2s-1} (x_i) = I\lidx{2s+1}(x_i^{-1}) \in \Lambda^\pm$. They have the following properties.

\begin{lemma}\label{lemma:iprop}
  \begin{enumerate}[(a)]
   \item \label{it:recurse}
   We have for every $s \in \nbb_0$,
   \begin{equation}\label{eq:irecurse}
      (-1)^{s} \sigma\lidx{2s+1} = I\lidx{2s+1} - \sigma\lidx{2} I\lidx{2s-1} +  \sigma\lidx{4} I\lidx{2s-3} - \ldots + (-1)^s \sigma\lidx{2s} I\lidx{1}.
   \end{equation}
   \item \label{it:ihom} $I_{2s+1}$ is homogeneous of degree $2s+1$, $s \in \zbb$.
   \item \label{it:igen} The $I_{2s+1}$ for $s \in \nbb_0$ ($s \in\zbb$) generate $\Lambda\inv$ ($\Lambda\inv^\pm$).
   \item \label{it:nonneg} For each $s$, the coefficients of $I_{2s+1}$ are all nonnegative.%
      \footnote{This motivates our sign convention for $I_{2s+1}$, which agrees with \cite{Christe:descendant} but differs from \cite{CardyMussardo:descendant,FringMussardoSimonetti:1993}.}
  \end{enumerate}
\end{lemma}

\begin{proof}
  Part (\ref{it:recurse}) follows by expanding the determinant \eqref{eq:idef} by the first column. From there, (\ref{it:ihom}) follows by induction for $s \geq 0$, and is then immediate for $s<0$.
  For part (\ref{it:igen}), first note that $I_{2s+1}\in\Lambda\inv$ ($s \geq 0$), which can be seen from (\ref{it:recurse}) by induction on $s$, using the relation $\alpha_y(\sigma_k) = \sigma_k - y^2 \sigma_{k-2}$. Now by Newton's identities, each $\sigma_k$ can be expressed as $\frac{1}{k}(-1)^{k-1} \psp{k}$ plus a polynomial in the $\psp{j}$, $j<k$. Using this repeatedly in \eqref{eq:irecurse}, we find that
  \begin{equation}\label{eq:piq}
      \psp{2s+1} = (-1)^s (2s+1) I_{2s+1} +Q_{s}(I_1,I_3,\dotsc,I_{2s-1},\psp{1},\psp{3},\dotsc,\psp{2s-1})
  \end{equation}
  with some polynomial $Q_{s}$. (Note that no even power sums occur on the r.h.s.\ due to Proposition~\ref{proposition:invargen}.)  Now applying \eqref{eq:piq} recursively, we can express every $\psp{2s+1}$ only in terms of the $I_{2s+1}$, showing that the $I_{2s+1}$ ($s \geq 0$) generate $\Lambda\inv$ by Proposition~\ref{proposition:invargen}. The proof for $\Lambda\inv^\pm$ is analogous. 
  Part (\ref{it:nonneg}) is a special case of a result by Kuipers and Meulenbeld \cite[Theorem 1]{KuiMeu:symmpoly}, applied on any fixed $\mathbb{R}^n$.
\cmpqed\end{proof}

Let us also consider the following combinations of the $I_{2s+1}$: for $s \geq 1$,
\begin{equation}\label{eq:jdef}
   J\lidx{2s+1} := \det \begin{pmatrix}
                             I\lidx{4s-1} & I\lidx{4s-3} & \dots & I\lidx{2s+1} \\ 
                             I\lidx{4s-3} & I\lidx{4s-5} & \dots & I\lidx{2s-1} \\ 
                               \vdots & & \ddots & \vdots \\
                             I\lidx{2s+1} & I\lidx{2s-1} & \dots & I\lidx{3}  
                             \end{pmatrix} \in \Lambda\inv.
\end{equation}
These functions are of interest because they vanish in less than $2s+1$ (but odd) dimensions.

\begin{lemma}
 \label{lemma:jprop}
  The $J\lidx{2s+1}$, $s \in \nbb$, enjoy the following properties: 
  \begin{enumerate}[(a)]
   \item \label{it:jprod} For $\xv \in \rbb^{2s+1}$,
   \begin{equation}\label{eq:jprod}
      J_{2s+1}(\xv) = \prod_{1 \leq i < j \leq 2s+1} (x_i + x_j).
    \end{equation}
   \item \label{it:jzero} For any $k<s$, and $\xv \in \rbb^{2k+1}$, we have $J\lidx{2s+1}(\xv) =  0$.     
  \end{enumerate}
\end{lemma}

\begin{proof}
  We first show (\ref{it:jzero}), where we can restrict ourselves to $k=s-1$. Noting that $\sigma\lidx{k} (\xv)=0$ for $\xv \in \rbb^{2s-1}$, $k \geq 2s$, the columns of the matrix in \eqref{eq:jdef} are linearly dependent by Eq.~\eqref{eq:irecurse}; hence the determinant vanishes.
  
  \sloppy Now for (\ref{it:jprod}), consider the rational function on $\cbb^{2s+1}$ given by $G(\xv):= J\lidx{2s+1}(\xv) \prod_{ i < j } (x_i + x_j)^{-1}$. Since $J\lidx{2s+1} \in \Lambda\inv$ and due to part (\ref{it:jzero}), the numerator vanishes where $x_i+x_j=0$; hence $G$ is analytic. On the other hand, both numerator and denominator are homogeneous of degree $s(2s+1)$; therefore $G$ is homogeneous of degree $0$, hence bounded, hence constant. One can fix this constant to be $1$ by direct computation. (For example, let $\xv=(\xv',y)\in\rbb^{2s}$; by repeated application of \eqref{eq:irecurse} one obtains $J\lidx{2s+1} (\xv) = y\sigma_{2s-1}(\xv')^2 J_{2s-1}(\xv') + O(y^2)$, from where the result follows by induction on $s$.)
\cmpqed\end{proof}

We now show a density property of $\Lambda\inv$ (and hence $\Lambda\inv^\pm$). 

\begin{proposition}\label{prop:idense}
  Let $n\in\nbb$, and for each $j= 1,\dotsc,n$, let $f_{j}$ be a continuous symmetric function from $\rbb^{j}$ to $\cbb$. For every $r>0$ and $\epsilon>0$, there exists a  $P\in\Lambda\inv$ such that
  \begin{equation}
     \big\lvert  P(e^{\thetav}) - f_{j} (\thetav) \big\rvert < \epsilon
     \quad 
     \text{for every $j \in \{1,\ldots, n\}$ and every $\thetav\in [-r,r]^{j}$}.
  \end{equation}
\end{proposition}

\begin{proof}
 Define the compact Hausdorff space
 \begin{equation}
    \xfk := \bigsqcup_{j=1}^n \big\{ \xv \in \rbb^{j} : e^{-r} \leq x_1 \leq x_2 \leq \dots \leq x_{j} \leq e^r \big\}.  
 \end{equation}
 With the obvious identification, we can consider $\Lambda\inv$ as a $\ast$-subalgebra (with identity) of $\ccal(\xfk,\cbb)$. We show that $\Lambda\inv$ separates points, i.e., if $\xv, \yv \in \xfk$ such that $P(\xv)=P(\yv)$ for all $P\in\Lambda\inv$, then $\xv=\yv$.  Let such $\xv \in \xfk\cap \rbb^i$, $\yv\in\xfk\cap\rbb^{j}$ be given. As $\psp{2k+1}\in\Lambda\inv$, we have in particular for all $k$,
 \begin{equation}\label{eq:psum}
     x_1^{2k+1} + \dotsb + x_{i}^{2k+1} = y_1^{2k+1} + \dotsb + y_{j}^{2k+1},
 \end{equation}
 and hence, noting $x_i \geq e^{-r}>0$,
 \begin{equation}
     \big(\frac{x_1}{x_i}\big)^{2k+1} + \dotsb + \big(\frac{x_{i-1}}{x_i}\big)^{2k+1} + 1= \big(\frac{y_1}{x_i}\big)^{2k+1} + \dotsb + \big(\frac{y_{j}}{x_i}\big)^{2k+1}.
 \end{equation}
 The left-hand side has a finite, nonzero limit as $k \to \infty$. For the right-hand side, this is true only if $y_{j}=x_i$. Hence we can cancel the last term on both sides of \eqref{eq:psum}. Continuing this scheme, we either arrive at $\xv=\yv$ (if $i=j$) or at a contradiction (if $i \neq j$).---Thus $\Lambda\inv$ separates points, and hence by the Stone-Weierstra\ss{} Theorem \cite[Ch.~V \S 8]{Conway:fa}, $\Lambda\inv$ is dense in $\ccal(\xfk,\cbb)$. After symmetric extension in the $j$ variables, and a variable transformation $\theta_i = \log x_i$, this is exactly the statement claimed.
\cmpqed\end{proof}

Further, we need estimates for functions in $\Lambda$ or $\Lambda^\pm$ and their derivatives. As in Sec.~\ref{sec:oddexamples}, we use the  notation $\partial_J = \partial/\partial\zeta_{j_1} \dotsm \partial/\partial\zeta_{j_n}$ for a set $J = \{j_1,\ldots,j_n\}\subset \nbb$. 
 
\begin{proposition}\label{proposition:pbound}
 Let $P \in \Lambda^\pm$. There exists $a,b>0$ such that for any $n \in \nbb$, any $\zetav\in\cbb^n$, and any set $J \subset \{1,\ldots,n\}$,
 \begin{equation}
    \big\lvert \partial_J P( e^{\zetav} ) \big\rvert \leq  a^n E(\re \zetav)^{b}.
 \end{equation}

\end{proposition}
\begin{proof}
 If $P,Q \in \Lambda^\pm$ fulfill an estimate of the type claimed, then so do $P+Q$, $cP$ (with $c \in \cbb$) as well as $P \cdot Q$, noting that the product rule reads
 \begin{equation}
    \partial_J (P\cdot Q) ( e^{\zetav} )=   \sum_{I \subset J} \partial_I P ( e^{\zetav} ) \; \partial_{J \backslash I} Q( e^{\zetav} )   
 \end{equation}
 and that the sum contains at most $2^n$ terms. It hence suffices to prove the statement for the generators $\psp{k}$, $k \in \zbb \backslash \{0\}$, which can be done by direct computation.
\cmpqed\end{proof}

\section{The functions \texorpdfstring{$M\odd_{n}$}{M[odd][n]}}\label{app:modd}

In this appendix we discuss the properties of the meromorphic functions $M\odd_{n}$ on $\cbb^{n}$ as introduced in Sec.~\ref{sec:oddexamples}, which are given by
%
\begin{equation} \label{eq:rprod}
     M\odd_{n}(\zetav)  := \prod_{ i < j } \tanh\frac{\zeta_i-\zeta_j}{2} = \prod_{i<j}\frac{e^{\zeta_i}-e^{\zeta_j}}{e^{\zeta_i}+ e^{\zeta_j}}. 
\end{equation}
As a first step, it is useful 
to rewrite the function using the following technique. 
By a \emph{pairing} of $n$ indices, we understand a set of pairs, $\mathfrak{p} = \{ (\ell_1,r_1),\dotsc,(\ell_k,r_k)\}$ where $k=\lfloor n/2 \rfloor$, where $\ell_j,r_j \in \{1,\dotsc,k\}$ are all pairwise different and $\ell_j < r_j$. We denote the set of all such pairings as $\pcal_{n}$, where $\pcal_{0} = \pcal_{1} = \{ \emptyset \}$. The \emph{signum} of a pairing $\mathfrak{p}$ is defined, in the case $n=2k+1$, as
\begin{equation}\label{eq:signPdef}
   \sign \mathfrak{p} :=  
   \sign \begin{pmatrix} 1 & 2 & 3 & 4 & \cdots & 2k\!-\!1 & 2k & 2k\!+\!1 \\
         \ell_1 & r_1 & \ell_2 & r_2 & \cdots & \ell_k & r_k &\hat{m} \end{pmatrix},
\end{equation}
where $\hat m$ is the unique number not occurring in the pairs; if $n=2k$, we drop the last column.  We note that this expression does not depend on the ordering of the pairs. Also, $\sign \emptyset := 1$.
With these definitions, we can express the function $M\odd_{n}$ as follows.

\begin{lemma} \label{lemma:rpairs}
For any $n \in \nbb_0$,
\begin{equation} \label{eq:rpairs}
M\odd_{n}(\zetav) =\sum_{\mathfrak{p} \in \pcal_{n}} \sign \mathfrak{p} \prod_{(\ell,r)\in \mathfrak{p}} \tanh \frac{ \zeta_{\ell} - \zeta_{r}}{2}.
\end{equation}
\end{lemma}

\begin{proof}
 Let us temporarily denote
 \begin{equation}
  T_{n}(\xv) := \sum_{\mathfrak{p} \in \pcal_n} \sign \mathfrak{p} \prod_{(\ell,r)\in \mathfrak{p}} \frac{x_\ell-x_r}{x_\ell+x_r};
 \end{equation}
 our aim is to show $M\odd_{n}(\zetav)=T_{n}(e^{\zetav})$. 
 Since $T_{n}$ is antisymmetric in its arguments, 
 the expression $T_{n}(\xv)\prod_{i<j}(x_i+x_j)$ is a skew-symmetric polynomial. Therefore, there exists 
\cite[Thm.~3.1.2]{Prasolov:polynomials} a symmetric polynomial $Q_n$ such that
\begin{equation}
     T_{n}(\xv) \prod_{i<j}(x_i+x_j) =  Q_n(\xv) \prod_{i<j}(x_i-x_j).
\end{equation}
But since $T_{n}$ is homogeneous of order 0, so is $Q_n$; thus $Q_n$ must be constant, and $T_{n}(e^{\zetav}) = Q_n M\odd_{n}(\zetav)$.
To determine the constant, note that
\begin{equation}
   \lim_{\epsilon\to 0} M\odd_{n}(\log(\epsilon),\log(\epsilon^2),\dotsc, \log(\epsilon^{n})) = 1 = \lim_{\epsilon\to 0} T_{n}(\epsilon,\epsilon^2,\dotsc,\epsilon^{n}).
\end{equation}
(All quotients $(x_\ell-x_r)/(x_\ell+x_r)$ etc.~converge to 1 in this limit; and one has $\sum_{\mathfrak{p}\in\pcal_n} \sign \mathfrak{p} = 1$, as can be seen by induction on $n$.) Thus $Q_n=1$, which concludes the proof.
\cmpqed\end{proof}

We collect the main features of the functions  $M\odd_{n}$.

\begin{proposition}\label{proposition:mprop}
For any fixed $n$, the functions $M\odd_{n}(\zetav)$ have the following properties:
\begin{enumerate}[(a)]
\item \label{it:manalytic} they are analytic where $|\im \zeta_i - \im \zeta_j|<\pi$ for all $1\leq i<j \leq n$;
\item \label{it:mantisymmperiod} they are totally antisymmetric, and $2\pi i$-periodic in each variable;
\item \label{it:mresidue} 
they have a first order pole at $\zeta_2 - \zeta_1 = i\pi$ with residue
\begin{equation}\label{eq:resM}
\res_{\zeta_2 - \zeta_1 = i\pi} M\odd_{n}(\zetav) = -2 M\odd_{n-2}(\zeta_3, \ldots,\zeta_{n}); 
\end{equation} 

\item\label{it:mcomplexbd} 
there is $c>0$ such that for all $\zetav \in \ical^{n}_\pm$,
\begin{equation}
  |M\odd_{n}(\zetav)| \leq c \operatorname{dist} (\im\zetav,\ical^{n}_\pm)^{-\lfloor n/2 \rfloor};
\end{equation}
\item \label{it:mrealbd} for all $\thetav \in \rbb^{n}$ they fulfill $\lvert M\odd_{n} (\pmb{\theta}) \rvert \leq 1$.
\end{enumerate}
\end{proposition}

\begin{proof}

Properties (\ref{it:manalytic}), (\ref{it:mantisymmperiod}) and (\ref{it:mrealbd}) can be read off directly from Eq.~\eqref{eq:rprod}.
For (\ref{it:mresidue}), one notes that the only factor in $M\odd_n$ contributing to the pole is $\tanh \frac{1}{2}(\zeta_1-\zeta_2)$, with residue $-2$; the claim then follows from \eqref{eq:rprod} or, alternatively, from \eqref{eq:rpairs}.

Regarding (\ref{it:mcomplexbd}), we estimate the function $M\odd_{n}(\zetav)$ for $\zetav \in \ical^{n}_+$ (the argument is similar for $\ical^{n}_-$). 
We remark that
\begin{equation}
  \Big\lvert \tanh\frac{\zeta}{2} \Big\rvert \leq c_1 \Big(1 + \frac{1}{|\zeta + i\pi|}\Big) 
  \leq \frac{5 c_1}{\im\zeta + \pi} \quad \text{for all }\zeta \in \rbb+i(-\pi,0)
\end{equation}
with some constant $c_1>0$. Applying this to every term in the representation \eqref{eq:rpairs}, we find a constant $c_2$ (depending on $n$) 
such that for all $\zetav \in \ical^{n}_+$,
\begin{equation}\label{eq:rkbound}
  |M\odd_{n}(\zetav)| \leq c_2 \Big( \max_{i<j} \frac{1}{\im(\zeta_i-\zeta_j) + \pi} \Big)^{\lfloor n/2 \rfloor}  \leq c_2 \operatorname{dist} (\im\zetav,\ical^{n}_+)^{-\lfloor n/2 \rfloor},
\end{equation}
noting that $\operatorname{dist}(\im\zetav,\ical^{n}_+) \leq  |\im(\zeta_i-\zeta_j)+\pi|$ for every $i,j$. 
\cmpqed\end{proof}

Finally, we derive a representation of $M\odd_{n}$ that is crucial for controlling its behaviour as an integral kernel.

\begin{lemma}\label{lemma:mlrsum}
 Let $\thetav\in\rbb^m$ and $\etav\in\rbb^n$. We have the identity
\begin{equation}\label{eq:mlrsum}
\begin{aligned}
   M\odd_{m+n}\ahponly{&}(\thetav,\etav+i\piv) = \ahponly{\\ &} \sum_{\substack{ k;(\ell_1,r_1),\dotsc,(\ell_k,r_k) \\ 1 \leq \ell_i \leq m < r_i \leq m+n}}
  \Big(\prod_{j=1}^k \coth \frac{\theta_{\ell_j}-\eta_{r_j-m}}{2}\Big)
  (-1)^{s_{\ell r}} M\odd_{m-k}(\hat \thetav) M\odd_{n-k}(\hat \etav).
\end{aligned}
\end{equation}
 The sum runs over all pairs of indices with the described properties, including over the number $k$ of pairs; it contains at most $2^{m+n}(\min(m,n)+1)!$ summands. $\hat\thetav\in\rbb^{m-k}$ denotes the $\theta_j$ with $j$ not in the list $\ell_1,\dotsc,\ell_k$, and $\hat\etav$ analogously.
 The integer $s_{\ell r}$ may depend on the choice of pairs.
\end{lemma}

\begin{proof}
 Recall that $M\odd_{m+n}$ is given as a sum over pairings as in \eqref{eq:rpairs}. Inserting $\zetav=(\thetav,\etav+i\piv)$, 
 we reorganize the sum over pairings $\mathfrak{p}$ as follows: We first fix the number $k$ of pairs $(\ell,r)\in\mathfrak{p}$ with $\ell \leq m$ and $r>m$, and sum over $k$; then we sum over all possibilities for such pairs at fixed $k$; then we sum over the possibilities for choosing the $\lfloor (m-k)/2 \rfloor$ pairs $(\ell,r)\in\mathfrak{p}$ with $\ell<r\leq m$, and the $\lfloor (n-k)/2 \rfloor$ pairs $(\ell,r)\in\mathfrak{p}$ with $m<\ell<r$, which complete the pairing of $m+n$ indices. For the last-mentioned two sums, applying \eqref{eq:rpairs} yields the factors $M\odd_{m-k}(\hat \thetav) M\odd_{n-k}(\hat \etav)$; the remaining factors of the product are of the form $\tanh((\theta_{\ell}-\eta_{r-m}-i\pi)/2)=\coth((\theta_{\ell}-\eta_{r-m})/2)$ with $\ell \leq m<r$.  
 Thus we arrive at Eq.~\eqref{eq:mlrsum}, where $s_{\ell r}$ is some integer depending on the pairing (which has no further relevance for us).

The sum contains at most $\binom{m}{k} \binom{n}{k} k! \leq 2^{m+n} k!$ summands at fixed $k$, so that the number $N$ of terms can be estimated by
\begin{equation}
   N \leq 2^{m+n} \sum_{k=0}^{\min(m,n)} k! \leq 2^{m+n} (\min(m,n)+1)!
\end{equation}
as claimed.
\cmpqed\end{proof}

\begin{acknowledgements}
We are grateful to C.J.\ Fewster for helpful comments. H.B.\ would like to thank the University of G\"ottingen and the TU M\"unchen for hospitality.
This work was supported by the Deutsche Forschungsgemeinschaft (DFG) within the Emmy Noether grants DY107/2-1 and CA1850/1-1. 
\end{acknowledgements}


\footnotesize
\bibliographystyle{ieeetr} 
\bibliography{../../integrable}

\end{document}